\documentclass [11pt]{article}
\usepackage{amsmath,amsthm,amsfonts,amscd,eucal,latexsym,amssymb,mathrsfs,cancel}
\usepackage{epsfig}  
\usepackage{simplewick,bm}
\usepackage{hyperref}
\usepackage{hyperref}
\usepackage{xcolor}
\usepackage{ulem}
\input xy
\xyoption{all}
 
\def\RR{{\mathbb R}}

\def\beq{\begin{eqnarray}}
\def\eeq{\end{eqnarray}}

\def\be{\beta}

\def\la{\lambda}

\def\om{\omega}
\def\ph{\varphi}

\def\supp{\textrm{supp}}

\newcommand{\xx}{\mathbf{x}}
\newcommand{\pp}{\mathbf{p}}
\newcommand{\ie}{{\it i.e.\ }}

\newtheorem{theorem}{Theorem}[section]

\newtheorem{proposition}{Proposition}[section]

\theoremstyle{definition}

\newcommand{\no}[1]{:\!#1\!:}

\begin{document}

%
%
 
\par 
\bigskip 
\LARGE 
\noindent 
{\bf Algebraic approach to Bose-Einstein Condensation in relativistic Quantum Field Theory.\\} 
{\bf \Large Spontaneous symmetry breaking and the Goldstone Theorem.\\} 
\bigskip 
\par 
\rm 
\normalsize 
 

\large
\noindent 
{\bf Romeo Brunetti$^{1,2,a}$}, {\bf Klaus Fredenhagen$^{3,b}$}, {\bf Nicola Pinamonti$^{4,5,c}$},  \\
\par
\small
\noindent$^1$ Dipartimento di Matematica, Universit\`a di Trento - Via Sommarive 14, I-38123 Povo (TN), Italy.
\smallskip

\noindent$^2$ Istituto Nazionale di Fisica Nucleare - TIFPA via Sommarive 14, I-38123 Povo (TN), Italy. 
\smallskip

\noindent$^3$ 
II. Institut f\"ur Theoretische Physik, Universit\"at Hamburg, Luruper Chaussee 149, D-22761 Hamburg, Germany.
\smallskip

\noindent$^4$ Dipartimento di Matematica, Universit\`a di Genova - Via Dodecaneso 35, I-16146 Genova, Italy. 
\smallskip

\noindent$^5$ Istituto Nazionale di Fisica Nucleare - Sezione di Genova, Via Dodecaneso 33, I-16146 Genova, Italy. 
\smallskip

\smallskip

\noindent E-mail: 
$^a$romeo.brunetti@unitn.it, 
$^b$klaus.fredenhagen@desy.de, 
$^c$pinamont@dima.unige.it\\ 

\normalsize

\par 
 
\rm\normalsize 

\rm\normalsize 
 
 
\par 
\bigskip 
 
\rm\normalsize 
\noindent {\small Version of \today}

\par 
\bigskip

\rm\normalsize

\bigskip

\noindent 
\small 
{\bf Abstract}
We construct states describing Bose-Einstein condensates at finite temperature for a  relativistic massive complex scalar field with $|\ph|^4$-interaction. 
We start with the linearised theory over a classical condensate and construct interacting fields by perturbation theory. Using the concept of thermal masses,
equilibrium states at finite temperature can be constructed by the methods developed in \cite{FredenhagenLindner} and \cite{DHP}. Here, the principle of perturbative agreement plays a crucial role.
The apparent conflict with Goldstone's Theorem is resolved by the fact that the linearized theory breaks the $U(1)$ symmetry, hence the theorem applies only to the full series but not to the truncations at finite order which therefore can be free of infrared divergences. 
 
\normalsize

\vskip .3cm

\section{Introduction}

In this paper we shall analyze the perturbative construction of a Bose-Einstein condensate for a relativistic charged scalar field theory at finite temperature. 

 The first experimental realization of Bose-Einstein condensation (BEC) in dilute vapours of alkali atoms has been obtained few years ago 
\cite{Anderson, Bradely, Davies}. These works have pushed a lot  the theoretical and experimental investigations of this phenomenon. 

Usually, Bose-Einstein condensation is discussed in the realm of non relativistic quantum theories.
(See e.g.
\cite{Pista} and references therein.) 
Bose-Einstein condensation in the non-interacting case is the phenomenon that below a certain critical temperature the ground state becomes macroscopically populated. A similar phenomenon is seen in systems with interaction where an analogous interpretation in terms of an effective single particle Hamiltonian can be made. The concept of BEC becomes mathematically precise by requiring that for an equilibrium state with average particle number $N$ the largest eigenvalue of its one-particle reduced density matrix is at least of order $N$ in the limit of large $N$. The eigenfunction to the highest eigenvalue of the reduced density matrix is then called the wave function of the condensate.
 
A basic fact about sufficiently dilute Bose gases is that the energy density is proportional to the square of the particle density, the proportionality factor being essentially the scattering length of the two-particle interaction potential, cf. Chapter 2 in \cite{LiebSeiringerSolovejYngvason}. By scaling the interaction potential so that the scattering length is proportional to the inverse of the particle number one arrives in the large $N$ limit at the Gross-Pitaevskii equation \cite{Gross,Pitaevskii} for the wave function of the condensate for Bosons that are confined in a trap, for instance a finite box \cite{LiebSeiringerYngvason, LiebSeiringerYngvasonCMP,LiebSeiringerSolovejYngvason}. This particular limit (``GP-limit'') is different from the thermodynamic limit where the density as well as the interaction potential is not scaled with $N$. 

Instead of considering a quantum system in a box and the limit where the box tends to cover the full space, we deal in this paper directly with the full spacetime.
We shall work in a relativistic quantum field theory setting. This presents some advantages, first of all the construction of the algebra of interacting observables in this case is free from infrared divergences\footnote{The infrared divergences at zero temperature in the nonrelativistic case have been discussed by Benfatto \cite{Benfatto} and by Di Castro group \cite{DiCastro}.}
  thanks to the causal property of the theory, see e.g. \cite{BrosBuchholz1, BuchholzRoberts}\footnote{The construction of an algebra of interacting nonrelativistic bosons was a longstanding open problem which was recently solved by Buchholz \cite{Buchholz-Interacting-Bosons} 
using the concept of the so-called resolvent algebra \cite{Buchholz-Grundling}}.
Furthermore, there are physical systems where the relativistic nature of fundamental physics is manifest together with the phenomenon of condensation. Here we have in mind the possible existence of Boson stars at cosmological level \cite{BosonStar},
and condensation phenomena in high energy physics as e.g. the quark gluon plasma 
\cite{QGP, Satz}. 

More precisely, we shall analyze various possible equilibrium states at  finite temperature for a complex scalar quantum field with mass $m$ and chemical potential $\mu$ both for free and self interacting theories.
In the free theory, if the mass $m$ is larger than $|\mu|$, there is a single equilibrium state for a given temperature with $n-$point functions described by tempered distributions. If the mass $m$ is smaller than $|\mu|$ there are no such states while if $m$ equals $|\mu|$ there are various equilibrium states.
These various states correspond to different phases of the system. Furthermore, the pure phases differ by macroscopic contributions, they have different one-point functions exhibiting spontaneous breakdown of the global $U(1)$ symmetry.
The charge density gets a finite contribution from the one-point functions. 
In the nonrelativistic limit the charged scalar field tends to the nonrelativistic scalar field and the states tend to the known equilibrium states of a nonrelativistic system of spinless non-interacting bosons in the thermodynamic limit. The charge density tends to the particle density, and the nontrivial one point function shows up in the long distance behavior of the 2-point function which coincides with the one-particle-reduced density matrix in the thermodynamic limit. 
A discussion about the equivalente of BEC with spontaneous breaking of gauge symmetry 
in the nonrelativistic setting, can be found e.g. in \cite{FPV,LiebSeiringerYngvason05,Suto}.

We then consider thermal equilibrium states in the case of a $\ph^4$ self interaction with positive coupling and chemical potential $\mu$. 
The traditional construction of non-zero temperature equilibrium states (KMS states) for perturbatively defined interacting quantum field theories 
suffers, even in the massive case, from spurious infrared divergences at higher loop order (see \cite{Altherr, Steinmann}). 
It was recently shown how these divergences can be circumvented \cite{Lindner, FredenhagenLindner}.
This construction amounts to an adaptation of formulas derived in rigorous statistical mechanics by Araki  \cite{Araki-KMS} to the framework of perturbative QFT. 
For $|\mu|<m$ the method works, and one obtains states which are invariant under the $U(1)-$symmetry.
 
For $|\mu|\ge m$ one expects spontaneous breakdown of symmetry. 
We thus expand the theory around a nontrivial solution $\phi$ of the classical field equation.  
In the limit of vanishing coupling constant $\lambda$ keeping the ratio $(\mu^2-m^2)/\lambda$ finite, 
the classical background tends to the condensate of the free theory. 
If one instead keeps $\mu$ and $m$ fixed, the classical 
 background dominates over the quantum fluctuations, actually $|\phi|^2$ diverges as $\lambda^{-1}$. 
 In particular, if one scales the charge density by multiplying with $\lambda$, one gets the charge density of the classical 
 solution rescaled by $\sqrt{\lambda}$ which is a solution of the field equation with $\lambda=1$. 
 This limit can thus be seen as an analogous of the GP limit for relativistic systems. 

Due to the Goldstone Theorem, one has to deal with massless modes and therefore with a slow decay of correlation functions.
In the case of a massless scalar field this problem could be circumvented by taking into account that the interaction produces at finite temperature a thermal mass. If this term is included in the free theory, the correlations of the unperturbed state decay sufficiently fast \cite{DHP}. In the case of BEC this is in conflict with the existence of a Goldstone mode induced by the spontaneous breakdown of the $U(1)$ symmetry of the model.

We solve this problem in the following way: We linearize the theory around the classical solution. The linearized theory breaks the $U(1)$ symmetry and shows nonvanishing thermal masses, hence the perturbative construction works as in the massless model treated in \cite{DHP}. The $U(1)$ symmetry is recovered for the full theory which then has a massless mode in agreement with the Goldstone theorem.

As mentioned before, we shall work in the relativistic quantum field theory setting called {\it perturbative algebraic quantum field theory} (pAQFT) \cite{BF00, BDF09, BFV03, HW01, HW02, FredenhagenRejzner}, see also the recent books on the subject \cite{Rejzner,Duetsch}.

The paper is organised as follows. In the next section we shall briefly recall  the framework of pAQFT, the key steps in the construction of the KMS state performed in \cite{FredenhagenLindner}  and few facts about the principle of perturbative agreement \cite{HW05} by which we can move the thermal mass into the free theory \cite{DHP}.
The third section contains the perturbative analysis of the massive complex scalar field with a $|\varphi|^4$ interaction, expanded around a solution with a nonvanishing condensate. 
We shall then discuss the construction of the interacting state at finite temperature over the condensate
and we analyze its adiabatic limit.
 The forth section contains the discussion of the formation of the condensate in connection to the spontaneous symmetry breaking. 
 We shall actually see that the symmetry is effectively broken in the background theory while it is recovered in the exact theory where it is thus spontaneously  broken.

\section{Equilibrium states for interacting quantum field theory}

In this section we briefly review the formalism of perturbative algebraic quantum field theory.
This framework combines renormalized perturbation theory with concepts of algebraic quantum field theory \cite{HaagKastler, Haag}. 
The basic step is an assignment of the algebra generated by the  observables of the theory to each region of the spacetime.   
Physical informations, like locality, is then stored in the relations between algebras labelled by different regions of the spacetime. 
In the case of an interacting theory treated with perturbation theory, the elements of the algebras associated to each region are given as formal power series in the coupling constant with values in suitable $*-$algebras. States are constructed in a second step as linear functionals on the algebra of observables which via the GNS construction then provide representations of the elements of the algebra by operators on a state space. Renormalization in this framework is automatically independent of the state; moreover, infrared problems do not occur in the construction of the algebra. They may become visible in the construction of states where they indicate physical properties of the system.

\subsection{Perturbative construction of the interacting quantum field theory and the adiabatic limit}\label{sec:quantum-algebras}

We shall here briefly recall the basic elements of the perturbative construction of the  $\phi^4$ scalar field theory propagating on a  four dimensional Minkowski spacetime $(\mathbb{M},\eta)$ where the metric $\eta$ has the signature $(-,+,+,+)$.
The Lagrangian is 
\begin{equation}\label{eq:lagrangian}
\mathcal{L} = -\frac{1}{2} \partial_\mu\phi\partial^\mu\phi - \frac{m^2}{2}\phi^2 - \frac{\lambda}{4} \phi^4  
\end{equation}
where $m$ is a positive mass and $\lambda$ the coupling to the non linear perturbation.

In a first step we construct the algebra corresponding to the free theory ($\lambda=0$).
We label the elements $O$ of the algebra by functionals on the classical configuration space which in our case is the space of smooth functions $\phi$ on Minkowski space,
\[
O[\phi]= \sum_{n=0}^N \int f_n(x_1,\dots, x_n) \phi(x_1)\dots \phi(x_n)d^{4n} x
\]
where $f_n$ is a compactly supported distribution on $\mathbb{M}^n$ which is symmetric under permutations of the arguments
and where we used the measure induced by the Minkowski metric. 
The spacetime support 
$\mathrm{supp}\,O$ of $O$ is the smallest closed subset $G$ of Minkowski space such that $\supp f_n\subset G^n$ for all $n$.   
$O$ is called regular if all $f_n$ are smooth, and it is local if all $f_n$ 
are of the form
\[
f_n(x_1,\dots,x_n)=P\prod_{i=2}^n\delta(x_1-x_i)
\]
with a partial differential operator $P$ with smooth compactly supported coefficients.

The product of regular observables is defined in terms of the commutator function $\Delta$, which is the retarded minus advanced fundamental solution of $(\square-m^2)\phi=f$. In the sense of generating functionals it is given by 
\[e^{i\phi(f)}\star e^{i\phi(g)}=e^{-\frac{i}{2}\langle f,\Delta g\rangle}e^{i\phi(f+g)}\]
where $\phi(f)=\int\phi(x)f(x)d^4x$ for a compactly supported test function $f$. 
Due to singularities of $\Delta$, the product cannot be extended to nonlinear local functionals as e.g. $\int f(x)\phi(x)^nd^4x$ with $n>1$ and a test function $f$.
This well known problem is, as usual, circumvented by replacing these functionals by so-called normal-ordered functionals. In our framework this means the following. 

Consider the Klein-Gordon equation with spacetime dependent mass $m(x)$ on a globally hyperbolic spacetime. Let $H$ be a symmetric bisolution of the form
\begin{equation}\label{eq:H}
H +\frac{i}{2}\Delta=  \lim_{\epsilon \to 0^+}\frac{U}{\sigma_\epsilon} + V\log\left(\frac{\sigma_\epsilon}{\xi^2}\right)+W
\end{equation}
where $U,V$ and $W$ are smooth symmetric functions of 2 spacetime points $x$ and $y$. $U$ and $V$ depend only on the geometry and on the mass near to the geodesic connecting the arguments, and
$\sigma_{\epsilon}(x,y) = \sigma(x,y) + i \epsilon (t(x) -t(y))$. $\sigma$ is the square of the geodesic distance between $x$ and $y$, equipped with the appropriate sign for spacelike and timelike separation, respectively, $t$ is a time function and $\xi$ a lenghtscale. Such a bisolution is called a Hadamard function \cite{KayWald}. According to Radzikowski \cite{Radzikowski}, a Hadamard function $H$ can be characterized as a symmetric bisolution with the property that the wave front set of $H +\frac{i}{2}\Delta$ satisfies a positivity condition (microlocal spectrum condition \cite{BFK}). Examples of Hadamard functions are the symmetric parts of the 2 point functions of vacuum and KMS states. 

Given a Hadamard function $H$, normal ordering is a linear map defined by
\[
e^{i\phi(f)}\mapsto \no{e^{i\phi(f)}}_H=e^{i\phi(f)}e^{\frac12\langle f,Hf\rangle}.
\]
The $\star$-product then is interwined with the so-called Wick product $\star_H$,
\[\no{e^{i\phi(f)}}_H\star \no{e^{i\phi(g)}}_H=\no{e^{i\phi(f)}\star_H e^{i\phi(g)}}_H\equiv\no{e^{i\phi(f+g)}}_H e^{-\langle f,(H+\frac{i}{2}\Delta)g\rangle}\ .\]
Due to the smaller wave front set of $H+\frac{i}{2}\Delta$ compared to $\Delta$, the Wick product can be extended to a larger class of functionals $O$.
The product among these objects is well defined as long as  
$\text{WF}(f_n) \cap  (\overline{V}_+^n\cup \overline{V}_-^n) = \emptyset$
where $\text{WF}(f_n)$ is the wave front set \cite{Hormander} of the distribution $f_n \in \mathcal{E}'(\mathbb{M}^n)$
 and  where $\overline{V}_\pm^n$ denotes the closure of the forward/backward light cones in $T^*{\mathbb{M}^n}$, namely
$
\overline{V}_\pm^n \doteq \{(x_1,\dots, x_n;k_1,\dots, k_n)\in T^*M^n | \langle k_i, \eta^{-1}k_i\rangle \leq 0, k_i^0 \geq 0  \}.
$
These functionals are called microcausal, and we denote their set by $\mathcal{F}_{\mu c}$. It contains in particular the local functionals, denoted by $\mathcal{F}_{\text{loc}}$, and their pointwise (classical) products, the multilocal functionals. We refer to \cite{BDF09, Rejzner} for further details on the definition of these sets. 
We can now extend the algebra of observables by normal-ordered microcausal functionals and define their product by
\[
\no{O_1}_H\star\no{O_2}_H=\no{O_1\star_H O_2}_H.
\]
Note that the enlarged algebra $\mathcal{A}$, obtained as $(\mathcal{F}_{\mu c},\star_H)$, does not depend on the choice of the Hadamard function. Only the labeling by the functional depends on $H$. Note furthermore that the normal-ordered functionals are, in general, no longer functionals on the configuration space due to the singularities of the Hadamard function. 

The standard normal ordering in Fock space with respect to annihilation and creation operators is obtained if one chooses the symmetrized vacuum 2 point function $\Delta_1$ as the Hadamard function. It has the nice feature that the vacuum expectation value $\omega_0(\no{O}_{\Delta_1})$ of a normal-ordered functional $\no{O}_{\Delta_1}$ coincides with the evaluation of the functional $O$ at $\phi=0$,
\[\omega_0(\no{O}_{\Delta_1})=O(\phi=0)\ .\]
A corresponding formula holds for any quasifree (Gaussian) state $\omega$, if one uses its symmetrized 2-point function as the Hadamard function, in particular for KMS states. 

This choice of normal ordering, however, is problematic when one wants to identify observables in different states or under the process of renormalization. Then a preferred Hadamard function with $W=0$ is better behaved, as first discussed by Hollands and Wald \cite{HW01}. In the case of a generic curved background that preferred function is in general not well defined but in the case of Minkowski space it was explicitly constructed in appendix A of \cite{BDF09}. It is unique up to the choice of a length scale.    

The field equation is not yet implemented into the algebra $\mathcal{A}$. The algebra contains instead an ideal generated by normal-ordered functionals which vanish on solutions. The so-called on shell algebra is obtained by taking the quotient with respect to this ideal. This quotient is faithfully represented on Fock space and coincides with the standard algebra of the free field. 

The off shell algebra $\mathcal{A}$, however, is better behaved under the time-ordered product which is used for the incorporation of interaction.       

\bigskip

Interacting fields can be constructed by means of causal perturbation theory, a method of renormalization elaborated by Epstein and Glaser \cite{EG} on the basis of ideas of Stueckelberg \cite{Stueckelberg,Stueckelberg2} and Bogoliubov \cite{BS}. It was further developed by Scharf and collaborators (see e.g. \cite{Scharf}).
It is also the basis for a treatment of interactions on curved spacetimes
\cite{BF00,HW01} where other versions of renormalization do not work.
Its main idea is the construction of time-ordered products of interaction Lagrangians.  In the work of Epstein and Glaser these products are operator valued distributions on Fock space. One uses the fact that the time-ordered product for noncoinciding points agrees with the operator product in the appropriate order. Renormalization then consists in extending these distributions to coinciding points. By induction
with respect to the number of factors one can show that this extension is always possible and unique up to the addition of a further interaction Lagrangian in each order. This corresponds precisely to the freedom in the choice of renormalization conditions known from other versions of reormalization.

In pAQFT one uses a version of causal perturbation theory which is independent of the choice of a state space. There the time ordering operator is a linear map $T$ from the algebra of multilocal functionals (with respect to the pointwise (classical) product)\footnote{The algebra of multilocal functionals with respect to the pointwise product is isomorphic to the symmetric tensor algebra of local functionals vanishing at $\varphi=0$. This fact has been proved in \cite{FredenhagenRejznerBV}.} to the algebra of $\mathcal{A}$. On regular functionals it is determined by 
\[Te^{i\phi(f)}=e^{i\phi(f)}e^{-\frac{i}{2}\langle f,\Delta_D f\rangle}=\no{e^{i\phi(f)}e^{-\frac{i}{2}\langle f,\Delta_F^H f\rangle}}_H\]
with the Dirac propagator $\Delta_D$ (the mean of retarded and advanced propagator) and the Feynman propagator 
$\Delta_F^H=\Delta_D-iH$ associated to the Hadamard distribution $H$. It satisfies the 
causal factorization condition
\[T(FG)=TF\star TG\]
if the support of $F$ does not intersect the past of the support of $G$.
To extend $T$ to multilocal functionals one fixes a Hadamard function $H$ and 
characterizes the extensions by the initial conditions
\[T1=1\text{ and }TF=\no{F}_H\]
for local functionals $F$ vanishing at $\phi=0$. The causal factorization condition then fixes the map $T$ on $n$-local functionals by its values on $k$-local functionals for $k<n$ up to a local functional.

In order to reduce the ambiguity in this extension we choose a Hadamard distribution with $W=0$ as discussed before, so only the dependence on the scale $\xi$ remains for local functionals.  

Given an interaction Lagrangian $\mathcal{L}_I$ where $\mathcal{L}_I(\varphi)$ is a translation invariant section of the jet bundle constructed over the field configuration $\varphi$, we consider the local functional 
\begin{equation}\label{eq:interaction-lagrangian}
V[\varphi] \doteq \lambda \int g \mathcal{L}_I(\varphi) d^4x
\end{equation}
with the test function $g$.
 
The formal $S-$matrix of the interaction Lagrangian $V$ can now be constructed as time-ordered exponential
\[
S(V) \doteq Te^{iV}
\]
in the sense of formal power series, hence $S(V)$ is an element of $\mathcal{A}[[\lambda]]$, the set of formal power series in the coupling constant $\lambda$ with coefficients in $\mathcal{A}$. The $\star$-product in $\mathcal{A}$ extends directly to $\mathcal{A}[[\lambda]]$. 
Relative $S-$matrices are then defined as
\[
S_V(F) \doteq S(V)^{-1} \star S(V+F), \qquad F\in \mathcal{F}_{\text{loc}}
\]
where the inverse is understood in terms of the $\star$ product, and the interacting fields are given in terms of the {\it Bogoliubov map} (also called M\o ller operator) which extracts the contributions of $S_V(\mu F)$ linear in $\mu$
\[
R_V(F) \doteq  -i\left.\frac{d}{d\mu} S_V(\mu F)  \right|_{\mu=0}
= S(V)^{-1} \star T(e^{iV}F). 
\]
Interacting observables can now be represented as elements of the algebra generated by the relative $S$ matrices $S_V(F)$
\[
\mathcal{A}_I(\mathcal{O}) = \left[ \left\{ S_V(F) | F\in \mathcal{F}_{\text{loc}}(\mathcal{O})  \right\}\right],
\]
{where the square brackets denote the set of linear combinations of products of the elements inside the brackets.}
We observe that the association of spacetime regions $\mathcal{O}$ to algebras $\mathcal{A}_I(\mathcal{O})$ forms a net of subalgebras of $\mathcal{A}[[\lambda]]$ in the sense of the Haag-Kastler axioms of algebraic quantum field theory. In the following we omit the symbol $\star$ for the product within $\mathcal{A}[[\lambda]]$ and replace it by juxtaposition.
\bigskip

The last step in constructing the interacting theory is the removal of the cutoff $g$ from the interaction lagrangian $V$ in \eqref{eq:interaction-lagrangian}. 
This can be done taking the {\it adiabatic limit} $g\to1$. At algebraic level, as soon as the interacting observables are supported on a compact region $\mathcal{O}$, 
this limit can be taken over larger and larger regions where the cutoffs are equal to $1$, further details can be found in \cite{BF00}.
Here we can make this construction more explicit making use of the time slice axiom and the causal properties of the $S-$matrix.
Actually, both the $S-$matrix and the relative $S-$matrix satisfy the following causal factorisation property valid for $A,B,C\in\mathcal{F}_{\text{loc}}$
\begin{equation}\label{eq:causal-factorisation}
S(A+B+C) = S(A+B)  S(B)^{-1}  S(B+C), \qquad A \gtrsim C
\end{equation}
where $A \gtrsim C$ means that $A$ is later than $C$ in the sense that $\supp(A)\cap J_{-}(\supp(C))= \emptyset$, where $J_\pm(\mathcal{O})$ denotes the causal future/past of $\mathcal{O}$.
This causal factorisation property implies that
\begin{equation}\label{eq:causalS1}
S_{V+W}(F) = S_V(F) ,   \qquad  W\gtrsim F 
\end{equation}
\begin{equation}\label{eq:causalS2}
S_{V+W}(F) = S_V(W)^{-1}  S_V(F)  S_V(W) ,  \qquad   F\gtrsim W\ .
\end{equation}
If $g,g'$ coincide on 	$J_+(\mathcal{O})\cap J_-(\mathcal{O})$ the correponding local functionals $V,V'$ differ by
\[V'-V=W_++W_-\]
with $\supp W_+\cap J_-(\mathcal{O})=\emptyset$ and  $\supp W_-\cap J_+(\mathcal{O})=\emptyset$,
hence
\[S_V(F)\mapsto S_{V'}(F)=S_V(W_-)S_V(F)S_V(W_-)^{-1}\text{ for }F\in\mathcal{F}_{\text{loc}}(\mathcal{O}\]
extends to an isomorphism $\mathcal{A}^{g}_I(\mathcal{O})\to\mathcal{A}^{g'}_I(\mathcal{O})$. 
The limit $g\to1$  can now be taken at the algebraic level.

\bigskip
Consider now a Cauchy surface $\Sigma = t^{-1}(0)$ where $t$ is the time coordinate of a standard Minkowski coordinate system which is fixed once and forever. An $\epsilon$ neighborhood of the Cauchy surface $\Sigma$ is  
\[
\Sigma_\epsilon \doteq \{ p\in M | t(p)\in (-\epsilon,\epsilon)\}\ .
\]
Interacting fields satisfy the {\it time slice axiom}, see \cite{ChilianFredenhagen}, namely, for every $A\in\mathcal{A}_{I}(\mathcal{O}) $ there exists a $C\in \mathcal{A}_{I}(\Sigma_\epsilon\cap J(\mathcal{O}))$ such that
\[
A=C+\no{W}
\] 
where $W\in\mathcal{F}_{\mu c}$ vanishes on solutions $\phi$. Here $J(\mathcal{O})=J_+(\mathcal{O}) \cup J_-(\mathcal{O})$.
Hence the on shell algebras $\mathcal{B}_I(\mathcal{O})$, obtained by taking quotients with respect to the ideal generated by these elements, are subsets of $\mathcal{B}_I(\Sigma_\epsilon)$.   
Thus to construct a state for the interacting algebra $\mathcal{B}_{I}(M)$ it suffices to construct it for $\mathcal{B}_{I}(\Sigma_\epsilon)$.
We therefore choose a cutoff function $g$ of the form 
\begin{equation}\label{eq:cutoff}
g(t,\mathbf{x}) = \chi(t) h(\mathbf{x})
\end{equation}
where now the time cutoff is realized by $\chi(t)$ which is a smooth function which is equal to $1$ for $t>-\epsilon$ and $0$ for $t\leq -2\epsilon$. Furthermore, $h$ is a space cutoff which is compactly supported on $\Sigma$.   
To obtain a state in the abiabatic limit, it is sufficient to consider
the limit where $h$ tends to $1$ keeping fixed the time cutoff $\chi$.

\subsection{Interacting KMS states and the adiabatic limit}\label{se:FL-KMS-state}
Equilibrium states are characterized by the Kubo-Martin-Schwinger (KMS)  condition, see \cite{KMS}. This condition yields canonical Gibbs states when they are well defined, but remains meaningful also for infinitely extended systems where the Gibbs formula can no longer be used \cite{KMS,Haag}. 

We recall here the definition. A state $\omega$ over a C*-algebra $\mathfrak{B}$ 
 satisfies the KMS condition with respect to the one parameter group of $*-$automorphisms $\tau_t$ at inverse temperature $\beta$  if for every $A,B\in\mathfrak{B}$ 
$\omega(A\tau_t B)$ is an analytic function for $\text{Im}(t)\in (0,\beta)$, continuous at the boundary and if 
\[
\omega(A\tau_{i\beta} B) =  \omega(BA). 
\]
A state which satisfies the KMS condition at inverse temperature $\beta$ is called $\beta$-KMS state.
It is automatically invariant under $\tau_t$ and satisfies similar relations for $n$-point functions.

If one uses the concept of KMS states for general $*$-algebras one has to enrich the definition by some of these properties. See e.g. Definition 1 in \cite{FredenhagenLindner} 
for an extended discussion.
  
Let $\tau_t$ denote the one parameter group of $*-$automorphisms of 
$\mathcal{B}$, the $*-$algebra of free fields, induced by the action of 
Minkowski space time translations 
\[
\tau_t(F)\doteq F_t\ ,\ F_t[\varphi] \doteq  F[\varphi_t], \qquad \varphi_t(s,\mathbf{x})\doteq \varphi(s-t,\mathbf{x}) \ .
\]
Consider now the following two-point function 
\begin{equation}\label{eq:2pt-free-KMS}
\omega_2^\beta(f,g) \doteq \frac{1}{(2\pi)^3}\int_{\mathbb{R}^4} \overline{\hat{f}(p) }\hat{g}(p)   \frac{1}{1-e^{-\beta p_0}}\delta(p^{2}+m^2) d^4p 
\end{equation}
The quasifree state $\omega^\beta$ constructed out of this two-point function for the free theory is a KMS state at inverse temperature $\beta$ with respect to time translations. This state is easily described by using normal ordering with respect to its symmetrization $H_{\beta}$ as
\begin{equation}\label{eq:free-KMS}
\omega^\beta(\no{F}_{H_{\beta}})= F[0], \qquad  F\in \mathcal{F}_{\mu c}\ . 
\end{equation}
The interacting time evolution $\tau_t^V$ in $\mathcal{B}_I(\mathcal{O})$ the subalgebra of $\mathcal{B}$ generated by $S(F)$ with $F$ supported in $\mathcal{O}$ which is a representation of the algebra of interacting observables supported in $\mathcal{O}$, is such that
\[
\tau_t^V(S_V(F)) \doteq S_V(F_t)
\]
whereas the free evolution is 
$\tau_t(S_V(F)) = S_{V_t}(F_t)$. 
To construct the KMS state for the interacting theory we have to relate the free and interacting time evolution. The causal factorisation property \eqref{eq:causal-factorisation} implies that 
\[
\tau^V_t({S_V(F)}) =  S_V(V_t-V)     \tau_t(S_V(F))    S_V(V_t-V)^{-1},\qquad  S_V(F)\in \mathcal{B}_I(\Sigma_{\epsilon}), t\geq 0.
\]
The map $t\mapsto U(t)\doteq S_V(V_t-V)$ from positive real numbers to unitary elements of $\mathcal{B}_I$ defines a cocycle which intertwines the free and interacting time evolution. The cocycle relation and its infinitesimal generators are 
\[
U(t+s) = U(t)\tau_t U(s), 
\qquad
H_I \doteq -i\left.\frac{d}{dt} U(t)\right|_{t=0} ,
\]
where, in the case of $V$ as in \eqref{eq:interaction-lagrangian} with $g$ as in \eqref{eq:cutoff} it turns out that
\[
H_I = \int h(\mathbf{x}) \mathcal{H}_I(\mathbf{x})    d^3\mathbf{x}, 
\qquad 
\mathcal{H}_I(\mathbf{x}) \doteq 
\int \dot\chi(t)  R_V(-\mathcal{L}_I(t,\mathbf{x})) dt.
\]
Hence $H_I$ and $\mathcal{H}_I$ play the role of the interacting Hamiltonian and the interacting Hamiltonian density.
We stress that due to the smearing in time, $\mathcal{H}_I(\mathbf{x})$ is a well defined formal power series with coefficients contained within the algebra of the free theory.

For any $\beta$-KMS state $\omega^\beta$ of the free theory, like 
the quasifree state \eqref{eq:free-KMS} constructed with the two-point function \eqref{eq:2pt-free-KMS} , and any spatial cutoff described by the test function $h$ on $\Sigma$ we obtain a $\beta-$KMS state of the theory with interaction $H_I(h)$
with respect to the evolution $\tau^V_t$ observing that 
\[
t\mapsto \omega^{\beta}(A   U(t))
\]
for every $A\in\mathcal{B}_I(\Sigma_\epsilon)$ can be analytically continued to $\text{Im}t\in[0,\beta]$. Hence, 
\begin{equation}\label{eq:FLstate}
\omega^{\beta, V}_h (A) \doteq \frac{\omega^{\beta}(A  U(i\beta))}{\omega^{\beta} (U(i\beta))}, 
\qquad 
A\in\mathcal{B}_I(\Sigma_\epsilon)
\end{equation}
defines a $\beta-$KMS state with respect to $\tau^V_t$, as proved in \cite{FredenhagenLindner}.

Furthermore, the expectation values in the state $\omega_h^{\beta,V}$ can be computed by the following formula
\begin{align}\label{eq:trunc}
\omega^{\beta,V}_h(A)=&\sum_{n}\int_{0\le u_1\le\dots u_n\le\beta} d u_1\dots d u_n\int_{\RR^{3n}}d^3\mathbf{x}_1\dots d^3\mathbf{x}_n 
h(\mathbf{x}_1)\dots h(\mathbf{x}_n)\nonumber\\
 &\omega_T^{\beta}\!\left(A;\tau_{iu_{1}}(\mathcal H_I(\mathbf{x}_{1}));\dots;\tau_{iu_n}(\mathcal H_I(\mathbf{x}_n))\right), 
 \qquad 
 A\in\mathcal{B}_I(\Sigma_\epsilon)
\end{align}
Here $\omega_T^\beta$ denotes the truncated functional associated to $\omega^\beta$. 

As shown in \cite{FredenhagenLindner}, the limit $h\to1$ 
can now be taken provided the truncated $n-$point functions decay sufficiently fast for large spatial separations. 
Furthermore, the obtained state does not depend on $\chi$ anymore. 
In this way one obtains the correlation functions for an interacting field in thermal equilibrium in the case of a massive theory.

\subsection{Principle of perturbative agreement	  -  Massless case}\label{se:perturbative-agreement}

If the linearized theory is massless, the limit $h\to1$ in \eqref{eq:trunc} cannot be taken  because the decay of the $n-$point function is too slow.
However, in the case of a $\phi^4$ theory, it is possible to use a similar construction \cite{DHP}. The idea is to modify the splitting
\begin{equation}\label{eq:splitting1}
\mathcal{L} = \mathcal{L}_0+\mathcal{L}_I, 
\qquad
\mathcal{L}_0\doteq-\frac{1}{2}\partial \phi \partial \phi, 
\qquad 
\mathcal{L}_I \doteq -\lambda\frac{1}{4} \phi^4
\end{equation}
by adding an artificial mass $M$ to the background theory and subtracting it in the interacting Lagrangian, namely 
\begin{equation}\label{eq:splitting2}
\mathcal{L} = \mathcal{L}_0'+\mathcal{L}_I', 
\qquad
\mathcal{L}_0'\doteq\mathcal{L}_0 - \frac{{M}^2}{2} \phi^2, 
\qquad 
\mathcal{L}_I' \doteq \mathcal{L}_I + \frac{{M}^2}{2} \phi^2.
\end{equation}
Let $H$ be the distinguished Hadamard function for $\square-M^2$ with a length scale $\xi$ and $T$ a time ordering operator with $TF=\no{F}_H$. 
Let $H_{\beta}=H+W_{\beta}$ denote the symmetrized 2-point function of the $\beta$-KMS state as in \eqref{eq:2pt-free-KMS} for the theory with the modified free Lagrangian $\mathcal{L}_0'$. Then 
\begin{equation}\label{eq:Termal-Massphi4}T(\phi^4)=\no{\phi^4}_{H_{\beta}}+6W_\beta(0,0)\no{\phi^2}_{H_\beta}+6W_{\beta}(0,0)^2\ .
\end{equation}
We see that the interaction Hamiltonian density in the KMS-state contains a mass term with a positive coefficient as long as $M^2<M^2_\beta$ with the thermal mass  
\[M^2_{\beta}=3\lambda W_{\beta}(0,0)\ .\] 
 
Under this condition the interaction Lagrangian remains convex and possesses a single stationary point at $\phi=0$. 
As discussed for example in \cite{DHP}, the thermal mass is 
\[
M_\beta^2 = \lambda\left(c_M M^2+
 \frac{1}{2\pi^2}\int_{0}^\infty   \frac{1}{e^{\beta \sqrt{p^2+M^2}}-1}  \frac{p^2}{\sqrt{p^2+M^2}} dp \right)
\]
where $c_M = \frac{1}{8\pi^2} \log(M\xi) $ is a renormalization constant and it depends on the length scales $\xi$ in \eqref{eq:H}. If $\xi M$ is equal to $1$ then $M_\beta^2$ vanishes in the limit $\beta\to\infty$ and $M_\beta^2 = T^2/12+O(M^2)$.
We finally observe that the theories constructed with the two different splittings are equivalent thanks to the {\it principle of perturbative agreement}, which has been shown to hold in \cite{HW05}, see also \cite{DHP}. Further details about the validity of this principle are collected in Appendix \ref{se:ppa}.
We finally recall that, if the interaction Lagrangian is quadratic in the field $Q=\int \delta m^2 \varphi^2 d^4x$ and if it corresponds to a perturbation of the mass $m$ of the free theory to $\sqrt{m^2+\delta m^2}$, the equilibrium state constructed as in \eqref{eq:FLstate} is the KMS state at inverse temperature $\beta$ with perturbed mass.
This last observation has been proved in Theorem 3 of \cite{Drago} and it shows that perturbative agreement is compatible with the construction of equilibrium states discussed in \cite{FredenhagenLindner}.

\section{Massive complex scalar field}

We discuss the equilibrium states at finite temperature $\beta$ with a nonzero chemical potential $\mu$.
We first discuss the free theory ($\lambda=0$), afterwards we study the corresponding states for the interacting theory.
We are interested in finding states which can be interpreted as exhibiting Bose-Einstein condensation.
The traditional way of defining BEC with particle numbers and occupation numbers cannot be applied in relativistic quantum systems. 
Instead we look at states with non vanishing one-point functions, thus showing spontanous breakdown of the internal $U(1)$ symmetry of the theory.
This symmetry is generated by a conserved current $J_{\mu}$. The charge density $J_0$,
\[
J_0(x)\doteq -i:\!\dot\ph^*\ph-\ph^*\dot\ph\!:_H \ ,
\]
replaces the particle density of the nonrelativistic theory.
The mean of the charge density then distinguishes between different phases.

In the free theory, at fixed inverse temperature $\beta$, there is a critical value for the mean charge density. Below this value the pure phases correspond to unique gauge invariant states, with a chemical potential $\mu$ depending on the charge density and with $\mu^2<m^2$.
If the charge density is above this threshold the chemical potential has to satisfiy $\mu^2=m^2$, the states corresponding to pure phases have nonvanishing one-point functions which are related by the gauge symmetry. (See {\cite{BRvol2}} for the concept of chemical potential in an algebraic formulation.)

Due to the non vanishing one-point function, the two-point function is not decaying at large separations. 
Similar non vanishing long distance correlations are the basis of the criterion for BEC in the non relativistic theory.
There one says that the ground state has a macroscopic occupation if the one-particle density matrix smeared in both entries over a spatial box of dimension $L$ grows at least as particle number $N$ in the limit where $L\to \infty$ keeping $N/L^3$ finite, see \cite{LiebSeiringerSolovejYngvason} for a more extensive discussion.

\subsection{Condensate in the free theory}\label{sec:free-theory}

Let us start discussing the condensate for a free massive complex scalar quantum field theory propagating in a Minkowski spacetime.
Let us denote by $\ph$ the associated field configuration. 
Its equilibrium states with {inverse temperature} $\beta>0$ and {\it chemical potential} $\mu$, $|\mu|<m$ are 
the states which 
satisfy the KMS condition with respect to the time evolution
\begin{equation}\label{time}
\tau_{t,\mu}(\ph(x))=\ph(x+te_0)e^{it\mu}
\end{equation}
where $e_0$ denotes the unit vector in time direction. 
The theory possesses an internal $U(1)$ symmetry which might be spontaneously broken in some of equilibrium states. 
Hence, an interesting observable to distinghish these states is the current density 
\begin{equation}\label{eq:free-current}
J_0(f)\doteq\int J_0(x)f(x) d^4x
\doteq-i\int \left(:\!\dot\ph^*\ph-\ph^*\dot\ph\!:_H \right)f d^4x
\end{equation}
where this is seen as an element of $\mathcal{A}$.
We observe that, in view of the symmetry of the Hadamard coefficients $U$,  $V_i$ of $V=V_n\sigma^n$ in \eqref{eq:H}, and $W_i$ of $W=W_n\sigma^n = (\omega_0-H-i\Delta/2)$,
$\omega_0$ being the two-point function of the vacuum state,
we have that 
$:\!\dot\ph^*\ph-\ph^*\dot\ph\!:_H = :\!\dot\ph^*\ph-\ph^*\dot\ph\!:_{H+W}.$
Hence, $J_0(f)$ can be seen as the current normal-ordered with respect to  vacuum state. 
The possible pure phases are thus characterized by the following proposition

\begin{proposition}
For inverse temperature $\beta>0$ and chemical potential $|\mu|<m$ there exists an unique KMS state with respect to $\tau_{t,\mu}$ whose $n-$point functions are tempered distributions. This state, denoted by $\om_{\beta,\mu}$ is quasi-free and
its two-point functions are 
\begin{equation}\label{eq:2pt-free}
\om_{\beta,\mu}(\ph^*(x)\ph(y))\doteq\frac{1}{(2\pi)^{3}}\int{d^4p}\,\delta(p^2+m^2)\epsilon(p_0)e^{ip(x-y)}\frac{1}{1-e^{-\beta(p_0-\mu)}}
\end{equation}
and $\om_{\beta,\mu}(\ph(x)\ph(y)))=\om_{\beta,\mu}(\ph^*(x)\ph^*(y)))=0$.
The charge density in this state is
\begin{align*}
\om_{\beta,\mu}(J_0(x))&=\int d^4p\,2|p_0|\delta(p^2+m^2)\left(\frac{1}{e^{\beta(|p_0|-\mu)}-1}-\frac{1}{e^{\beta(|p_0|+\mu)}-1}\right).
\end{align*}
It holds that 
\begin{equation}\label{eq:rho-critical}
|\om_{\beta,\mu}(J_0(x)) | \leq \rho_{cr}(\beta)\doteq 
\om_{\beta,m}(J_0(x))
\end{equation}
where $\rho_{cr}(\beta)$ is the critical charge density.
For $\beta>0$ and $\mu= \pm m$ there exist various KMS states with respect to $\tau_{t,\mu}$. Let us denote by $\Omega_{\beta,\pm m}$ the set of quasifree KMS states.
The pure phases are the extremal points in $\Omega_{\beta,\pm m}$ and these states are
\[
\omega_{\beta,c}^\pm = \om_{\beta,\pm m}\circ\gamma_c^{\pm}
\]
where $\gamma_c^{\pm}$ is an automorphism which is generated by
\begin{equation}\label{eq:aut-gamma}
\gamma^{\pm}_c(\ph(x))=\ph(x)+e^{\pm ix^0m}c(\mathbf{x})
\end{equation}
where $c$ is a harmonic function of the spatial variables $\mathbf{x}$, $\Delta c=0$, with the spatial Laplacian $\Delta$.
In this cases 
\begin{align*}
\om_{\beta,c}^{\pm}(J_0(x))&=\int d^4p\,2|p_0|\delta(p^2+m^2)\left(\frac{1}{e^{\beta(|p_0|\mp m)}-1}-\frac{1}{e^{\beta(|p_0|\pm m)}-1}\right)\pm 2m|c|^2.
\end{align*}
Furthermore,
\[
|\om_{\beta,c}^{\pm}(J_0(x))| \geq \rho_{cr}(\beta).
\]
\end{proposition}
\begin{proof}
First of all we observe that the KMS states corresponding to pure phases are quasifree states with at most a non trivial one-point function. A proof of this fact can be found in \cite{RoccaSirugueTestard}.
Furthermore, the truncated two-point function are constrained by the KMS condition to be equal to \eqref{eq:2pt-free}.
The one-point function $\om(\varphi)$ is constrained by the equation of motion and by request of invariance under the action of $\tau_{t,\mu}$. In particular, invariance under the action $\tau_{t,\mu}$ implies that the function
\[
t\mapsto \omega(\varphi(x))e^{-i\mu t}
\] 
is constant in time. This function needs to be a solution of $-\Delta + m^2 -\mu^2$, however, for $|\mu|<m$ these solutions cannot be tempered distributions. 
The inequalities involving the critical charge density $\rho_{cr}(\beta)$ is an immediate consequence of the form of the expectation value of $J_0$ in the analyzed states.
\end{proof}

The state $\omega_{\beta,\mu}$ respects the $U(1)-$symmetry of the theory. 
For chemical potentials $\mu=\pm m$ there exist many equilibrium states at fixed temperature and the $U(1)$ symmetry is spontaneously broken in the states $\om_{\beta,c}^{\pm}$  for $c\neq0$.
At zero temperature, all the states with chemical potential $|\mu|<m$ coincide with the vacuum. 
Hence, the vacuum expectation value 
of the charge density $J_0(x)$ vanishes in that limit. 
Thus any non vanishing charge density in the limit $T=\beta^{-1}\to0$ of vanishing temperature requires a condensate. 

 At finite $\beta$ and $\mu=m$ we have  
$
\om_{\beta,c}^{+}(J_0(x))
= \rho_{cr}(\beta)
+ 2m|c|^2$.
A non vanishing condensate can occur only if $|\om_{\beta,c}^{+}(J_0)|>\rho_{cr}(\beta)$.
Since $\rho_{cr}(\beta)>0$ is monotonically decreasing in $\beta$, diverges for $\beta\to0$ and tends to $0$ for large $\beta$,  
at a fixed charge density $\om_{\beta,c}^{\pm}(J_0)$, there is a critical temperature $T_{cr}>0$ such that only for $T<T_{cr}$ a condensate can be formed.

In Appendix \ref{sec:non-rel-limit} we shall compute the nonrelativistic limit of the states analyzed in this section showing that
the charged scalar field tends to the nonrelativistic scalar field and the states tend to the known equilibrium states of a nonrelativistic system of spinless non-interacting bosons in the thermodynamic limit. Furthermore, the charge density converges to the particle density. Finally, we see that 
the  nontrivial one point function shows up in the long distance behavior of the 2-point function which coincides with the one-particle-reduced density matrix in the thermodynamic limit.

\subsection{Massive complex scalar field with $\varphi^4$ interaction over the condensate}\label{se:Lagrangians}

In this section we start discussing the perturbative construction of the $\varphi^4$ interacting theory over a suitable classical solution of the equation of motion which represents the condensate in the Minkowski spacetime. 
The Lagrangian of the theory we are considering is thus
\[
\mathcal{L}  =  -\frac{1}{2} \partial \overline{\varphi}\partial {\varphi} - \frac{1}{2}m^2 |\varphi|^2 - \frac{\lambda}{4} |\varphi|^4
\]
where $\varphi$ is a complex scalar field. Following a similar procedure presented in section III of \cite{Schmitt}, we expand $\mathcal{L}$ around a real classical solution $\phi$ which represents the condensate. 
Hence
\begin{equation}\label{eq:field-decomposition}
\varphi = e^{-i\mu x^0}(\phi + \psi)
\end{equation}
where $\mu$ is again the chemical potential, $x^0$ is a fixed Minkowski time and $\psi$ is a complex scalar field which describes the perturbations. Its real and imaginary parts are denoted by $\psi_1$ and $\psi_2$ and thus
\[
\psi = \psi_1+i\psi_2.
\]
The Lagrangian density can now be written as a sum of contributions homogenous in the number of fields $\psi$ as follows
\[
\mathcal{L} = \mathcal{L}_0+\mathcal{L}_2+\mathcal{L}_3+\mathcal{L}_4 
\]
where 
\begin{align*}
\mathcal{L}_0 &= \frac{1}{2} \left|(\partial_0-i\mu)\phi\right|^2  - \frac{1}{2} |\nabla \phi|^2 - \frac{\lambda}{4}|\phi|^4 - \frac{1}{2}m^2|\phi|^2 \\
\mathcal{L}_2 &= \frac{1}{2} \left|(\partial_0-i\mu)\psi\right|^2  - \frac{1}{2} |\nabla \psi|^2 -\lambda \phi^2|\psi_1|^2
-\frac{1}{2}(\lambda\phi^2+m^2)|\psi|^2 \\
\mathcal{L}_3 &=  -\lambda \phi \psi_1 |\psi|^2 \\
\mathcal{L}_4 &=  -\frac{\lambda}{4}|\psi|^4.
\end{align*}
The term $\mathcal{L}_1$ vanishes because $\phi$ is chosen to be a  stationary point for the classical action $\int \mathcal{L}_0 d^4x$.
In the following, we shall choose a non vanishing $\phi$ to describe the condensate, we discuss the quantization of the linearized theory ($\mathcal{L}_2$) and finally we use perturbation theory over the linearized theory to take into account $\mathcal{L}_3+\mathcal{L}_4$. 

Contrary to the case of the free theory, in the interacting theory the chemical potential is not restricted to the interval $[-m,m]$. A chemical potential outside of this interval induces a spontaneous breakdown of symmetry showing up in a non-vanishing one-point-function and, as a consequence, in long range behavior of the two-point-function, similar to the non-relativistic case. In contrast to the free case, states with different condensates are not in mutual thermal equilibrium, since their chemical potentials differ.

\subsubsection{The condensate in the vacuum theory}\label{sse:condensate}

We look for the case of a translation invariant background $\phi$.
Then, 
the kinetic term in $\mathcal{L}_0$ has no effect and $\phi$ is a stationary point for 
\[
I=\int U(|\phi^2|) d^4x
\]
where  
\[
U(|\phi^2|) =    -\frac{\lambda}{4}|\phi|^4-\frac{1}{2}(m^2-\mu^2)|\phi|^2 
\]
hence, it holds 
\begin{equation}\label{eq:bec}
 |\phi|^2 =  \frac{\mu^2-m^2}{\lambda}
\end{equation}
and only one real, positive  and translational invariant background solution $\phi$ is thus available for $\mu^2>m^2$.
We notice that for fixed $\mu^2>m^2$ the background value of the field $\phi$ is of order $1/\sqrt{\lambda}$.
In this case, we observe that $\mathcal{L}_2$ does not depend on $\lambda$, $\mathcal{L}_3$ is of order $\sqrt{\lambda}$ while $\mathcal{L}_4$ is of order $\lambda$.  
In the next we shall construct the interacting field theory with perturbation methods considering $\mathcal{L}^I=(\mathcal{L}_3+\mathcal{L}_4)$ the interaction Lagrangian. Hence, the solution we shall obtain will be a formal power series in $\sqrt{\lambda}$.

In the next we shall discuss the construction of the quantum theory over the background discussed so far.

We argue that there exists a limit in which all the correlation functions are dominated by the classical 
background $\phi$. 
Actually, in the limit $\lambda\to0$ keeping $|\mu|^2-m^2$ finite, the classical background $\phi$ diverges as $\lambda^{-1/2}$, furthermore, the linearized theory is not affected by changes of $\lambda$ while 
 the $S-$matrix constructed with the interacting Lagrangian tends to $1$. 

Hence, 
we expect that, under this limit, the one-point function rescaled by $\sqrt{\lambda}$ tends to the background value 
$\tilde{\phi} e^{-i m x^0}$ where $\tilde{\phi} = \sqrt{\lambda} \phi$, and similarly,  
the rescaled charge density $\lambda J_0$ tends to the charge density of the background
$2\mu|\tilde{\phi}|^2$. Both these quantities do not depend on $\lambda$.
We finally observe that, the rescaled background $\tilde{\phi}$, is a solution of the equation of motion descending from the rescaled classical Lagrangian density
\[
\tilde{\mathcal{L}}_0 = 
\lambda \mathcal{L}_0 = 
\frac{1}{2} |(\partial_0-i\mu)\tilde{\phi}|^2  - \frac{1}{2} |\nabla \tilde{\phi}|^2 - \frac{1}{4}|\tilde{\phi}|^4 - \frac{1}{2}m^2|\tilde{\phi}|^2
\]
which is also independent on $\lambda$.

This is in analogy to what happens in the nonrelativistic case, actually there under the Gross-Pitaevskii limit, the density of the ground state tends to the density of a suitable classical solution of the Gross-Pitaevskii equation \cite{LiebSeiringerYngvason, LiebSeiringerYngvasonCMP}, see in particular Theorem 1.1 and Theorem 1.2 in \cite{LiebSeiringerYngvason}. 

We thus argue that, the equation of motion corresponding to the rescaled zeroth order Lagrangian $\tilde{\mathcal{L}}_0$ can be interpreted as an analogous of the Gross-Pitaevskii equation in the relativistic setting and thus the limit $\lambda \to 0$ taken with $m$ and $\mu$ fixed, can be understood as the analogous of the Gross-Pitaevskii limit discussed in the introduction.

\subsubsection{Linearized theory}
The first step to construct the quantization of $\varphi$ is the analysis of the linearized equations of motion for the fluctuations $(\psi_1,\psi_2)$ around $\phi$. They have the form
\begin{equation} \label{eq:linearized-equations}
\begin{aligned}
(\square-M_1^2)\psi_1-2\mu \dot{\psi}_2&=0 \\
(\square-M_2^2)\psi_2+2\mu \dot{\psi}_1&=0 
\end{aligned}
\end{equation}
where 
\begin{equation}\label{eq:masses}
M_1^2=(m^2-\mu^2) + 3\lambda \phi^2 
\qquad\text{and}\qquad
M_2^2= (m^2-\mu^2)+ \lambda \phi^2.
\end{equation}
Notice that if \eqref{eq:bec} holds,
$M_1^2=2(\mu^2-m^2) $  and $M_2^2=0$.
Hence, we assume $M_1>M_2\geq0$. Let us introduce 
\begin{align*}
M^2 =\frac{M_1^2+M_2^2}{2},
\qquad
\delta M^2 =\frac{M_1^2-M_2^2}{2}.
\end{align*}
We observe that the equations \eqref{eq:linearized-equations} for $\tilde{\psi}=(\psi_1,\psi_2)$ can be written in a compact form 
$
D \tilde\psi = 0
$
where $D$ is given in terms of the standard Pauli matrices 
\[
\bm\sigma_1=\begin{pmatrix}
0& 1\\1& 0
\end{pmatrix}, \qquad
\bm\sigma_2=\begin{pmatrix}
0& -i\\i& 0
\end{pmatrix}
,\qquad
\bm\sigma_3=\begin{pmatrix}
1& 0\\0& -1
\end{pmatrix}.
\]
as
\begin{equation}\label{eq:D-form}
\begin{aligned}
D &= (\square-M^2) \mathbb{I} - \delta M^2 \bm\sigma_3 - i 2\mu \partial_0 \bm\sigma_2 
\\
\overline{D} &= (\square-M^2) \mathbb{I} + \delta M^2 \bm\sigma_3 + i 2\mu \partial_0 \bm\sigma_2.
\end{aligned}
\end{equation}
Notice that 
\[
D\overline{D}= \left((-\square+M^2)^2-(\delta M^2)^2 + 4\mu^2 \partial_0^2\right) \mathbb{I}.
\]
The retarded and advanced propagators of the theory can be obtained as
\[
\Delta_R \doteq \overline{D} (D\overline{D})_R,\qquad
\Delta_A \doteq \overline{D} (D\overline{D})_A,
\]
where $(D\overline{D})_R$ and $(D\overline{D})_A$ are the retarded and advanced fundamental solutions of $D\overline{D}$.
Let us thus study 
\begin{align*}
\widehat{D\overline{D}} &= \left( (p^2+M^2)^2 - (\delta M^2)^2 -4\mu^2p_0^2 \right) \mathbb{I}  \\
& = \left( p_0^4      -2p_0^2  (\mathbf{p}^2+M^2 +2\mu^2)   +   (\mathbf{p}^2+M^2)^2 - (\delta M^2)^2\right) \mathbb{I}
\end{align*}
Hence, the four solutions of $p_0^4      -2p_0^2  (\mathbf{p}^2+M^2 +2\mu^2)   +   (\mathbf{p}^2+M^2)^2 - (\delta M^2)^2=0$ are $\pm\omega_\pm$ where 
\begin{align}
\omega_\pm^2 &= w^2 +2\mu^2 \pm \sqrt{(w^2 +2\mu^2)^2  - w^4 +(\delta M^2)^2} 
\notag\\
  &= w^2 +2\mu^2 \pm \sqrt{4\mu^4 + 4\mu^2w^2  +(\delta M^2)^2} 
  \notag\\
  &= w^2 +2\mu^2 \pm \sqrt{(w^2 +2\mu^2)^2  - w_1^2w_2^2}
  \label{eq:omegapm}
\end{align}
where now $w^2 \doteq \mathbf{p}^2+M^2$ and $w_i^2\doteq\mathbf{p}^2+M^2_i$.

We notice that if $M_2=0$ 
we have that $w_2=0$ for 
$|\mathbf{p}|=0$
 and thus
\[
\lim_{|\mathbf{p}| \to 0 }\omega_-^2   =   0
\]
hence a massless mode is present in this system as expected by the Goldstone theorem. 
However, if the linearized theory is not in a ground state, it could happen that the
normal-ordered interaction Lagrangian with respect to the state, as in \eqref{eq:Termal-Massphi4},
contains quadratic terms that
could contribute to the masses of the fluctuations.
If we use the formula
\[
\prod_i  \frac{1}{x-x_i} = \sum_i \frac{1}{x-x_i} \prod_{j\neq i} \frac{1}{x_i-x_j}
\]
valid for pairwise different $x_1,\dots, x_n$, a couple of times, we get
\begin{align*}
\widehat\Delta_R 
&= \frac{\widehat{\overline{D}}}{(\omega_+^2-\omega_-^2)} 
\left(  \frac{1}{(p_0+i\epsilon)^2-\omega_+^2} 
-\frac{1}{(p_0+i\epsilon)^2-\omega_-^2}\right)
\end{align*}
where recalling \eqref{eq:D-form}
\begin{equation}\label{eq:Dhat}
\widehat{\overline{D}} = -(p^2+M^2) \mathbb{I} + \delta M^2 \bm\sigma_3 + 2\mu p_0 \bm\sigma_2.
\end{equation}
We can construct $\Delta_A$ just changing $i\epsilon\to - i\epsilon$, while the Feynman propagator $\Delta_F$ is obtained substituting 
$(p_0+i\epsilon)^2$ with $p_0^2+i\epsilon$ and multiplying by $i$.
Finally, the  commutator function is 
\[
\widehat\Delta = \frac{2\pi i \widehat{\overline{D}}}{\omega_+^2-\omega_-^2} \epsilon{(p_0)}\left( \delta(p_0^2-\omega_+^2)-\delta(p_0^2-\omega_-^2) \right).
\]
With $\Delta$ at disposal the quantum product can be given as in section \ref{sec:quantum-algebras}, in this way we obtain the $*-$algebra of field observables $\mathcal{A}_0$. 
The analog of the Hadamard singularity $H$ \eqref{eq:H} for this theory can be given. The form of some of the corresponding Hadamard coefficients are discussed in Appendix \ref{sec:hadamard-coefficient}. 
The extended $*-$algebra of field observables $\mathcal{A}$ containing Wick polynomials normal-ordered with respect to $H$ is obtained as in \ref{sec:quantum-algebras}. 
  
\subsection{KMS states for the linearized theory}\label{se:kms-linear}

In view of the decomposition of the field $\varphi$ given in \eqref{eq:field-decomposition} the action of $\tau_t$ on $\psi$ as time translation is equivalent to the action of $\tau_{t,\mu}$ on $\varphi$ as given in \eqref{time}.
Hence, having the causal propagator of the linearized theory at disposal, we can construct the two-point function of the quasifree $\beta-$KMS state with respect to time translation $\tau_t$ of the $\psi$ fields
as 
\[
\widehat{\omega}_{\beta,\psi} = \frac{i\widehat\Delta}{1-e^{-\beta p_0}}.
\]
Introducing 
\[
\mathcal{S} \doteq\begin{pmatrix} \psi_1(x)\psi_1(y) &  \psi_1(x)\psi_2(y)\\\psi_2(x)\psi_1(y)& \psi_2(x)\psi_2(y)
\end{pmatrix}
\]
we have that the two-point function of the quasifree $\beta-$KMS state $\omega_{\beta,\psi}$ is in position space
\begin{equation}\label{eq:2pt-condensate}
\omega_{\beta,\psi}(\mathcal{S})=
\frac{1}{(2\pi)^{3}} \int d^4p\; e^{ip(x-y)} 
\frac{  \epsilon{(p_0)}}{\omega_+^2-\omega_-^2}  
\left( \delta(p_0^2-\omega_+^2)-\delta(p_0^2-\omega_-^2) \right)
\frac{(-\widehat{\overline{D}})}{1-e^{-\beta p_0}}.
\end{equation}
Recalling the form of $\omega_\pm $ in \eqref{eq:omegapm} we notice that if $M_1>M_2>0$  
\[
\omega_\pm^2=w^2 +2\mu^2 \pm \sqrt{(w^2 +2\mu^2)^2  - w_1^2w_2^2}  >0, \qquad
\omega_+^2-\omega_-^2 =2\sqrt{4\mu^4 + 4\mu^2w^2  +(\delta M^2)^2}  >0
\]
this means that no infrared divergences are present in $\omega_{\beta,\psi}$ if $M_2>0$.
The two-point function of the ground state of the $\psi_i$ theory (keeping the condensate $\phi\neq 0$) can be obtained taking the limit $\beta\to\infty$ of \eqref{eq:2pt-condensate}.
Hence, to study expectation values in the state $\omega_{\beta,\psi}$ of observables normal-ordered with respect to the vacuum $\omega_{\infty,\psi}$ we consider $W=\omega_{\beta,\psi}-\omega_{\infty,\psi}$ and we obtain
\begin{align}
W(\mathcal{S}) 
\label{eq:W}
&=
\frac{1}{(2\pi)^{3}} \int d^4p\; e^{ip(x-y)} 
\frac{1}{\omega_+^2-\omega_-^2}  
\left( \delta(p_0^2-\omega_+^2)-\delta(p_0^2-\omega_-^2) \right)
\frac{(-\widehat{\overline{D}})}{e^{\beta |p_0|}-1}.
\end{align}
We observe that in the coinciding point limit, the off diagonal expectation values are vanishing 
\[
W(\psi_1(x)\psi_2(x))=0, \qquad W(\psi_2(x)\psi_1(x))=0
\]
and, introducing  $2\delta\omega^2\doteq\omega_+^2-\omega_-^2 = 2\sqrt{4\mu^4 + 4\mu^2 w^2  +(\delta M^2)^2}$ we have that 
 $W(\psi_i^2)$, the coinciding point limits of the diagonal elements of $W(\mathcal{S})$,  are
\begin{align} \label{eq:psi21}
W(\psi_1^2)
& =
\frac{1}{(2\pi)^{3}} \int d^3\mathbf{p}\;  
  \left( \frac{\delta\omega^2+2\mu^2+\delta M^2 }{\delta\omega^2}
  \frac{1}{2\omega_+}
  \frac{1}{e^{\beta \omega_+}-1}+
\frac{
\delta\omega^2 -2\mu^2-\delta M^2 
}{\delta\omega^2}\frac{1}{2\omega_-}\frac{1}{e^{\beta \omega_-}-1} \right)
\\
\label{eq:psi22}
W(\psi_2^2)
& =
\frac{1}{(2\pi)^{3}} \int d^3\mathbf{p}\;  
  \left( \frac{\delta\omega^2+2\mu^2-\delta M^2 }{\delta\omega^2}
  \frac{1}{2\omega_+}
  \frac{1}{e^{\beta \omega_+}-1}+
\frac{
\delta\omega^2 -2\mu^2+\delta M^2 
}{\delta\omega^2}\frac{1}{2\omega_-}\frac{1}{e^{\beta \omega_-}-1} \right)
\end{align}
Notice that the integrand in both $W(\psi_1^2)$ and $W(\psi_2^2)$ are positive.

To analyze some properties of the condensate in the linearized theory we compare the expectation values of the current density
in the state $\omega_{\beta,\psi}$
with \eqref{eq:free-current}, namely the current density of the  free theory analyzed in section \ref{sec:free-theory}.
To this end, we recall the decomposition \eqref{eq:field-decomposition} and we get that 
\begin{equation}\label{eq:J-linear-theory}
J_0=-i (:\dot{\overline{\ph}}{\ph}-{\overline{\ph}}\dot{\ph}:_H)
= \tilde{j}  -i \phi (\dot{\overline{\psi}}-\dot{\psi})+2\mu :|\phi+\psi|^2 :_H
\end{equation}
where now
\[
\tilde{j}=-i  \left(:\dot{\overline{\psi}}\psi -\overline{\psi}\dot\psi :_H\right) = 2\left( :\dot\psi_1\psi_2-\psi_1\dot\psi_2:_H\right)
\]
and $H$ is the distinguished 
Hadamard function constructed in \eqref{eq:H-D}.
We furthermore observe that, up to some choice of the renormalization freedom, $:\dot\psi_1\psi_2-\psi_1\dot\psi_2:_H=:\dot\psi_1\psi_2-\psi_1\dot\psi_2:_{\omega_{\infty,\psi}}$
and $:|\psi|^2:_{H}=:|\psi|^2:_{\omega_{\infty,\psi}}=:\psi_1^2:_{\omega_{\infty,\psi}}+:\psi_2^2:_{\omega_{\infty,\psi}}$. Hence 
\begin{align} \label{eq:psi2}
 \omega_{\beta,\psi}(:|\psi|^2:_H) 
& =\frac{1}{(2\pi)^{3}} \int d^3\mathbf{p}\;  
  \left( \left(1+ \frac{2\mu^2 }{\delta\omega^2}\right)
  \frac{1}{\omega_+}
  \frac{1}{e^{\beta \omega_+}-1}+
\left(1-\frac{2\mu^2}{\delta\omega^2}\right)\frac{1}{\omega_-}\frac{1}{e^{\beta \omega_-}-1} \right)
\end{align}
and
\begin{align} \label{eq:tildej}
\omega_{\beta,\psi}(\tilde{j})
& =
\frac{4\mu}{(2\pi)^{3}} \int d^3\mathbf{p}\;  
\frac{1}{\delta \omega^2}  
\left( 
\frac{\omega_-}{e^{\beta \omega_-}-1}
-
\frac{\omega_+}{e^{\beta \omega_+}-1}
\right).
\end{align}
Hence, 
\[
\omega_{\beta,\psi}(J_0) = \omega_{\beta,\psi}(\tilde{j}) +2\mu \; \omega_{\beta,\psi}(:|\psi|^2:_H)   +2\mu|\phi|^2.
\]
Notice that $\omega_+^2>\omega_-^2$ and that 
$\delta \omega^2 \geq 2 \mu^2$, hence
the integrand in $\omega_{\beta,\psi}(:|\psi|^2:_H) $  given in \eqref{eq:psi2} is always positive and monotonically decreasing in $\beta$. Similarly, for positive $\mu$, the integrand in $\omega_{\beta,\psi}(\tilde{j})$ given in \eqref{eq:tildej} is also always positive and monotonically decreasing in $\beta$. Finally, both expressions  
\eqref{eq:psi2} and \eqref{eq:tildej}
 are diverging for $\beta\to 0$  and vanishes for $\beta\to\infty$. Hence
 similarly to the discussion given in section \ref{sec:free-theory}
 we have that 
 $2\mu|\phi^2|$ plays the role of the condensate charge density.
 
Consider now the case where $M_i$ are given in \eqref{eq:masses} with $\phi$ chosen to satisfy \eqref{eq:bec}.
In this case $\lambda |\phi|^2 = \mu^2-m^2$ and thus the linearized theory does not depend on $\lambda$ while the background field scales as $\lambda^{-1}$. The charge density is thus dominated by the charge density of the background $2\mu|\phi|^2$ thus confirming that the limit $\lambda\to0$ taken with fixed $\mu$ and $m$ is the relativistic analogous of the Gross-Pitaevskii limit discussed in the introduction and at the end of section \ref{sse:condensate}. 

Following closely the discussion given at the end of section \ref{sec:free-theory},
we also see that in this case the critical charge density equals $\rho_{cr}(\beta)$ given implicitly in \eqref{eq:rho-critical}.
Finally, in the limit $\lambda\to0$ taken keeping the ratio $(\mu^2-m^2)/\lambda$ finite we have that the states $\omega_{\beta,\psi}$ of the 
linearized theory discussed so far tend to $\omega^\pm_{\beta,c}$ with $c=\phi$.

\subsubsection{Thermal masses}

Having analyzed the equilibrium state of the free theory on field observables $\mathcal{A}$, the next step in the construction of an equilibrium state for the interacting theory will be an application of the analysis given in \cite{FredenhagenLindner} and summarized in section \ref{se:FL-KMS-state},  namely to use \eqref{eq:FLstate} starting with a quasifree state whose two-point function is given in \eqref{eq:2pt-condensate}. 
However, we expect that the limit $h\to1$ cannot be directly taken because, as discussed above, if \eqref{eq:bec} holds, the mass $M_2$ given in \eqref{eq:masses} vanishes, hence, 
for vanishing spatial momentum, $\omega_-^2$ is also vanishing.
This implies that, the various propagators of the linearized theory diverge for $p\to0$. 
Hence, in agreement with Goldstone theorem a massless mode is present in this case. This implies a slow decay in the connected $n-$point functions constructed with $\omega_{\beta, \mu}$ given in \eqref{eq:2pt-condensate}. 

In order to cure this problem we use a different splitting of the Lagrangian into the free and interacting part. Actually, we add a virtual mass $m_v^2$ to the linearized fields and we remove them in the interaction Lagrangian. More precisely, the Lagrangian of the free theory is now
\begin{equation}\label{eq:L2prime}
\mathcal{L}_2' = \mathcal{L}_2 - \frac{m_v^2}{2}|\psi|^2 
\end{equation}
while the modified interaction Lagrangian is 
\[
{\mathcal{L}'}^I = \mathcal{L}^I +  \frac{m_v^2}{2}|\psi|^2 
 =
 \mathcal{L}_3+ \mathcal{L}_4+  \frac{m_v^2}{2}|\psi|^2.
\]
The elements of the interacting algebra are now given  
in terms of two parameters $\lambda$ and $m_v$. 
More precisely, keeping $\mu$ fixed, as in \eqref{eq:bec}, they are formal power series in $\sqrt{\lambda}$ with coefficients depending on $ m_v^2$, which can be understood as a partial resummation of the original perturbative expansion. 
The advantage of this new expansion is in the fact that the coefficients remain finite in the adiabatic limit, when they are evaluated in the state representing the condensate at finite temperature.
We furthermore observe that the principle of perturbative agreement discussed below implies that the final theory does not depend on this extra parameter $m_v$. 

Let $H_{\psi,\beta}$ be the symmetrized two-point function of the $\beta$-KMS given in \eqref{eq:2pt-condensate},
we observe that if $m_v$ is chosen to be sufficiently small, the interaction Lagrangian normal-ordered with respect to $H_{\psi,\beta}$ is again convex. 
To see this in detail, let $T$ be a time ordering operator such that  $TF = \no{F}_H$ where $H$ is the distinguished 
Hadamard function constructed in \eqref{eq:H-D}.
We have up to a choice of renormalization freedom (the lenghtscale $\xi$ in \eqref{eq:H-D} chosen in such a way that  
$:|\psi|^4:_H=:|\psi|^4:_{H_{\infty,\psi}}$ where $H_{\infty,\psi}$ is the symmetrized two-point function of the vacuum obtained taking the limit $\beta\to\infty$ in \eqref{eq:2pt-condensate}) 
\begin{equation}\label{eq:modification}
T\left( \frac{1}{4}|\psi|^4\right) =
\frac{1}{4}\no{|\psi|^4}_{H_\beta}+
\frac{1}{2}  (3 m_{\beta,1}^2 +     m_{\beta,2}^2) \no{|\psi_1|^2}_{H_\beta}
+
\frac{1}{2}  (3 m_{\beta,2}^2 +     m_{\beta,1}^2) \no{|\psi_2|^2}_{H_\beta}
+C
\end{equation}
where $C$ is a constant which can de discarded and the two thermal masses $m_{\beta,i}$ have been computed above in \eqref{eq:psi21} and \eqref{eq:psi22}
\[
m_{\beta,1}^2 \doteq W(\psi_1^2), \qquad m_{\beta,2}^2 \doteq W(\psi_2^2).
\]
hence $T(\frac{1}{4}|\psi|^4 -\frac{m_v^2}{2}|\psi|^2)$ remains convex, provided
$m_v^2< \lambda (3 m_{\beta,1}^2 +     m_{\beta,2}^2) $ and $m_v^2 < \lambda (3 m_{\beta,2}^2 +     m_{\beta,1}^2)$.
For this reason it is expected that the stability properties of the theory are not altered 
adding the virtual masses $m_v$ in the free theory.

\subsection{Condensate and perturbative agreement}

We need to check that the Wick monomials in the 
interaction Lagrangian 
originally constructed over the linearized theory $\mathcal{L}_2$, are not corrected because of the new splitting. 
In other words, we prove that the principle of perturbative agreement holds also when a condensate is present. 
Let us recall the form of the equation of motion for $\tilde\psi = (\psi_1,\psi_2)$ given in \eqref{eq:D-form}
\begin{align*}
D\tilde\psi = 0, \qquad D = (\square-M^2) \mathbb{I} - \delta M^2 \bm\sigma_3 - i 2\mu \partial_0 \bm\sigma_2
\end{align*}
and consider the preferred Hadamard function $H_{M^2,\delta M^2,\mu}$, with a lengthscale $\xi$, associated with this operator constructed in Appendix \ref{sec:hadamard-coefficient}.

We prove now that the time ordering operator $T_{M,\delta M,\mu}(F) = \no{F}_{H_{M^2,\delta M^2,\mu}}$ satisfies the principle of perturbative agreement. To this end consider the $2\times 2$ matrix $\Psi=\{\psi_i\psi_j\}_{i,j\in\{1,2\}}$,
following the discussion presented in Appendix \ref{se:ppa}
we want to prove that 
\[
\Delta \Psi =  \gamma T_{0,0,\mu}\Psi - T_{M,\delta M,\mu} \Psi
\]
vanishes, where $\gamma$ is the map which intertwines $T_{0,0,\mu}(\psi_i(x)\psi_j(y))$ to $T_{M,\delta M,\mu}(\psi_i(x)\psi_j(y))$. 
Formally, indicating with the supscript $c$ the quantities referred to the condensate $(M,\delta M,\mu)$ and with the supscript $0$ those referred to the vacuum $(0,0,m)$ we have to compute
\begin{align*}
\Delta \Psi & = \lim_{y\to x}\left((H_F^0(x,y))_{\text{ren}} -  H_F^c(x,y)\right)
\end{align*}
where $H_F^{c/0}$ are the time-ordered/Feynman propagator associated to the Hadamard functions $H^{c/0}$.
By power counting we notice that all the contributions larger than order two in $(D^c-D)$ are removed from 
$\Delta \Psi$ by renormalization.
In order to check if there is a finite reminder after this renormalization we analyze the form of the Hadamard singularity $H_{M^2,\delta M^2,\mu}$ given in Appendix \ref{sec:hadamard-coefficient}, we remove the contributions of order lower than the third in $x_i$ from $H_{M^2+x_1,\delta M^2+x_2,\mu+x_3}$ before computing the coinciding point limit.
Let us recall the form of some Hadamard coefficient given in Appendix \ref{sec:hadamard-coefficient}. From equations \eqref{eq:U} and \eqref{eq:V0-1} we have 
\[
U =\cos(\mu x^0)\mathbb{I} - i\bm\sigma_2 \sin(\mu x^0)
\]
and
\[
V_0 = -\frac{1}{2} U \left(   (\mu^2+M^2) \mathbb{I} +
\delta M^2
\left(\frac{\sin(2\mu x^0)}{2\mu x^0}\bm\sigma_3 
+\frac{\cos(2\mu x^0)-1}{2\mu x^0} \bm\sigma_1\right) \right).
\]
hence, in $H_{M^2+x_1,\delta M^2+x_2,\mu+x_3}$,
$U/\sigma$ does not depend on $x_1$ and $x_2$ and the contributions in $x_3$ larger than the second order vanish in the coinciding point limits.
Similarly, the contributions larger than second order vanish also in $V\log(\frac{\sigma}{\xi^2})$.
As for the mass perturbations we thus have that the Wick monomials $\Psi^n$ computed with respect to the Hadamard parametrix do not change under the action of the map which intertwines the time ordering constructed with two different sets of parameters $M,\delta M, \mu$.

\subsection{Construction of the condensate, cluster estimates }

To construct the state at finite temperature over the condensate, we follow the construction given in \cite{FredenhagenLindner} and summarized in section \ref{se:Lagrangians}. 
In particular, at fixed spatial cutoff $h$, the equilibrium state at inverse temperature $\beta$ can be constructed as in \eqref{eq:trunc}. Using the spatial translation invariance of the interacting hamiltonian and denoting by $\tau_{t,\mathbf{x}}$ the $*$-automorphisms realizing a spacetime translation of step $(t,\mathbf{x})$ we have, for any element $A$ of $\mathcal{A}_I({\Sigma_\epsilon})$,
\begin{align}
\omega^{\beta,V}_h(A) \doteq&\sum_{n}\int_{0\le u_1\le\dots u_n\le\beta} d u_1\dots d u_n\int_{\RR^{3n}}d^3\mathbf{x}_1\dots d^3\mathbf{x}_n 
h(\mathbf{x}_1)\dots h(\mathbf{x}_n)\nonumber\\
 &\omega^\beta_T\!\left(A;\tau_{iu_{1},\mathbf{x}_{1}}(\mathcal H_I(0));\dots;\tau_{iu_n;\mathbf{x}_{n}}(\mathcal H_I(0))\right)
 \label{eq:KMS-h-finite}
\end{align}
where $\omega^\beta_T$ denotes the truncated $n-$point function of the state $\omega_{\beta,\psi}$.
Hence, in order to discuss the limit $h\to1$ we need to control the decay for large spatial directions of the truncated $n-$point functions.
We have actually the following theorem

\begin{theorem}\label{th:cluster} (Cluster expansions).
Consider $A_i\in\mathcal{A}(\mathcal{O})$ where $\mathcal{O}\subset B_R$ the open ball of radius $R$ centered at the origin of the Minkowski spacetime and
\[	
F(u_1,\mathbf{x}_1;\dots ;u_n,\mathbf{x}_n)\doteq \omega_T(A_0 ; \tau_{iu_1,\mathbf{x}_1}(A_1);\dots ;\tau_{iu_n,\mathbf{x}_1}(A_n)).
\]
There exists a constant $C$ such that 
\[
|F(u_1,\mathbf{x}_1;\dots ;u_n,\mathbf{x}_n)| \leq C e^{-\frac{m}{\sqrt{n}}r},\qquad 
r=\sqrt{\sum_i |\mathbf{x}_i|^2}
\]
for $r>4cR$, uniformly in $u$ for $0 < u_1 < \dots < u_n <\beta$  with
$\beta-u_n \geq \frac{\beta}{n+1}$.
\end{theorem}
\begin{proof}
Thanks to the decay property for large spatial separations of the locally smeared two-point functions given in Proposition \ref{pr:decay},  
the proof of this theorem can be done in a similar way as the proof of Theorem 3 of \cite{FredenhagenLindner}. We recall here the main steps of that proof and we adapt them to the case studied here.

The truncated $n-$point functions can be written as a sum over all possible connected graphs joining $n$ points. 
We shall denote the set of connected graphs, without tadpoles, with $n+1$ vertices $V=\{0,\dots, n\}$ as $\mathcal{G}_{n+1}^c$. Furthermore, for any $G\in\mathcal{G}_{n+1}$, $E(G)$ denotes the set of edges of $G$. For any $l\in E(G)$, $s(l)$ and $r(l)$ denote the source and the range of $l$. A graph $G$ is considered to be 
in $\mathcal{G}_{n+1}$ only if, for every $l$, $s(l)<r(l)$. 
Finally $l_{ij}(G)$ is the number of lines connecting $i,j\in V$.
With these definitions
\[
F(u_1,\mathbf{x}_1;\dots ;u_n,\mathbf{x}_n)\doteq \sum_{G\in\mathcal{G}^c_{n+1}} \frac{1}{{\text{sym}(G)}}  F_G(u_1,\mathbf{x}_1;\dots ;u_n,\mathbf{x}_n)
\]
where $\text{sym}(G)=\prod_{i<j} l_{ij}(G)!$ is a numerical factor and  
\[
F_G(u_1,\mathbf{x}_1;\dots ;u_n,\mathbf{x}_n)  \doteq \left( \prod_{0\leq i<j\leq n} \Gamma^{ij} \right) \left.(A_0 \otimes \tau_{iu_1,\mathbf{x}_1}(A_1)\otimes\dots \tau_{iu_n,\mathbf{x}_1}(A_n))\right|_{(\psi^0,\dots,\psi^n) = 0}.
\]
Furthermore, 
\[
\Gamma^{ij} = \int  d^4x d^4y \;\mathcal{K}(x-y) \frac{\delta}{\delta \psi^i(x)}   \otimes \frac{\delta}{\delta\psi^j(y)}
\] 
with the integral kernel $\mathcal{K}(x-y) = \omega_{\beta,\psi}(\psi(x)\psi(y))$, given in terms of the thermal two-point function of the background theory \eqref{eq:2pt-condensate}. Furthermore,
$\psi^j=\psi^j_1+i\psi^j_2$ is the field configuration in the $j-$th factor of the tensor product and the functional derivative $\frac{\delta}{\delta\psi^j}$  acts on the $j-$th factor of the tensor product.
We have that 
\[
F_G(U,\mathbf{X})  \doteq  \int dP
\left(\prod_{l\in E(G)}  e^{p_l^0(u_{s(l)}-u_{r(l)})}
e^{i \mathbf{p}_l (\mathbf{x}_{s(l)}-\mathbf{x}_{s(l)})}
\hat{\mathcal{K}}(p_l) \right)\hat\Psi(-P,P)
\]
where $U=(u_0,\dots, u_n)$, $\mathbf{X}=(\mathbf{x}_0, \dots ,\mathbf{x}_n)$ with $u_0=0$ and $\mathbf{x}_0=0$, while $P=(p_1,\dots ,p_{|E(G)|})$  and
\[
\Psi(Z,Y) = \left.\left(\prod_{l\in E(G)} \frac{\delta}{\delta \psi^{s(l)}(z_l)}\otimes \frac{\delta}{\delta \psi^{r(l)}(y_l)} \right)
(A_0 \otimes A_1\otimes\dots A_n)\right|_{(\psi^0,\dots,\psi^n) = 0}.
\]
We observe that
\[
\hat{\mathcal{K}}(p) =  
\left(\lambda_+(p) + \lambda_-(p)\right) (-\widehat{\overline{D}})
\]
where $\lambda_+$ and $\lambda_-$ are the positive and negative frequency part
\begin{align*}
\lambda_+(p) &= \frac{1}{\omega_+^2-\omega_-^2} \left( \frac{\delta(p_0-\omega_+)}{2\omega_+}- \frac{\delta(p_0-\omega_-)}{2\omega_-}\right) \frac{1}{1-e^{-\beta p_0}}\\
\lambda_-(p) &= -\frac{1}{\omega_+^2-\omega_-^2} \left( \frac{\delta(p_0+\omega_+)}{2\omega_+}- \frac{\delta(p_0+\omega_-)}{2\omega_-}\right) \frac{1}{1-e^{-\beta p_0}}.
\end{align*}
Hence, separating the positive and negative contributions in $F_G$ we get
\begin{align*}
F_G(U,\mathbf{X})  =  \sum_{P_2(E(G))}\int d\mathbf{P}
\left(\prod_{l_+\in E_+(G)}
e^{p_{l_+}^0(u_{s(l_+)}-u_{r(l_+)})}
e^{i \mathbf{p}_{l_+} (\mathbf{x}_{s(l_+)}-\mathbf{x}_{r(l_+)})}  \lambda_+(p_{l_+})(-\hat{\overline{D}}(p_{l_+}))\right)\\
\cdot\left(\prod_{l_-\in E_-(G)}
e^{p_{l_-}^0(u_{s(l_-)}-u_{r(l_-)})}
e^{i \mathbf{p}_{l_-} (\mathbf{x}_{s(l_-)}-\mathbf{x}_{s(l_-)})}  \lambda_-(p_{l_-})(-\hat{\overline{D}}(p_{l_-}))\right)
\hat\Psi(-P,P)
\end{align*}
where the sum is taken over all possible partitions of $E(G)$ in up to two sets $\{E_+(G),E_-(G)\}\in P_2(E(G))$.
We proceed now splitting again these contributions over the two possible frequencies $\omega_{\pm}$. Hence, denoting by
\begin{align*}
\lambda_{++}(p) &\doteq \frac{1}{\omega_+^2-\omega_-^2} \left( \frac{\delta(p_0-\omega_+)}{2\omega_+}\right) \frac{1}{1-e^{-\beta p_0}}\\
\lambda_{+-}(p) &\doteq \frac{-1}{\omega_+^2-\omega_-^2} \left( \frac{\delta(p_0-\omega_-)}{2\omega_-}\right) \frac{1}{1-e^{-\beta p_0}}\\
\lambda_{-+}(p) &\doteq -\frac{1}{\omega_+^2-\omega_-^2} \left( \frac{\delta(p_0+\omega_+)}{2\omega_+}\right) \frac{1}{1-e^{-\beta p_0}}\\
\lambda_{--}(p) &\doteq \frac{1}{\omega_+^2-\omega_-^2} \left(\frac{\delta(p_0+\omega_-)}{2\omega_-}\right) \frac{1}{1-e^{-\beta p_0}}
\end{align*}
we have
\begin{align*}
F_G(U,\mathbf{X})  =  \sum_{P_2(E(G))}
\sum_{P_2(E_+(G))}\sum_{P_2(E_-(G))}\int d{P}
\left(\mathcal{Q}_{++} \cdot \mathcal{Q}_{+-} \cdot\mathcal{Q}_{-+}\cdot\mathcal{Q}_{--}\right)
\hat\Psi(-P,P)
\end{align*}
where
\begin{align*}
\mathcal{Q}_{\sigma\sigma'}
&\doteq
\left(\prod_{l\in E_{\sigma\sigma'}(G)}
e^{p_{l}^0(u_{s(l)}-u_{r(l)})}
e^{i \mathbf{p}_{l} (\mathbf{x}_{s(l)}-\mathbf{x}_{r(l)})}  \lambda_{\sigma \sigma'}(p_{l})(-\hat{\overline{D}}(p_{l}))\right)
\qquad \sigma,\sigma'\in \{+,-\}.
\end{align*}
The function $\hat\Psi$ is an entirely analytic function which grows at most polynomially in every direction. We might thus 
integrate over all possible $p_0$ to get
\begin{align*}
F_G(U,\mathbf{X})  \doteq  \sum_{P_2(E(G))}
\sum_{P_2(E_+(G))}\sum_{P_2(E_-(G))}\int d\mathbf{P}
\left(\tilde{\mathcal{Q}}_{++} \cdot \tilde{\mathcal{Q}}_{+-} \cdot\tilde{\mathcal{Q}}_{-+}\cdot\tilde{\mathcal{Q}}_{--}\right)
\Phi(\mathbf{P})
\end{align*}
where now 
\[
\Phi(\mathbf{P}) \doteq\left.\hat\Psi(-P,P)\right|_{p_0^{l_{\sigma\sigma'}} = \sigma\omega_{\sigma'}(\mathbf{p}_l)} \]
 and
\begin{align*}
\tilde{\mathcal{Q}}_{\sigma\sigma'}&=
(\sigma1) (\sigma' 1)\left(\prod_{l\in E_{\sigma\sigma'}(G)}
\frac{e^{ \sigma\omega_{\sigma'}(\mathbf{p}_{l})(u_{s(l)}-u_{r(l)})}
e^{i \mathbf{p}_{l} (\mathbf{x}_{s(l)}-\mathbf{x}_{r(l)})}  
}{(\omega_+^2-\omega_-^2)2\omega_{\sigma'}}
\frac{(-\hat{\overline{D}}(\sigma\omega_{\sigma'},\mathbf{p}_{l}))}{1-e^{-\sigma\beta \omega_{\sigma'}}}
\right)\\
&= 
(\sigma' 1)
\left(\prod_{l\in E_{\sigma\sigma'}(G)}
\frac{e^{-\left(\frac{(1-\sigma)}{2}\beta + \sigma(u_{r(l)}-u_{s(l)})\right)\omega_{\sigma'}}
e^{i \mathbf{p}_{l} (\mathbf{x}_{s(l)}-\mathbf{x}_{r(l)})}  
}{(\omega_+^2-\omega_-^2)2\omega_{\sigma'}}
\frac{(-\hat{\overline{D}}(\sigma\omega_{\sigma'},\mathbf{p}_{l}))}{1-e^{-\beta \omega_{\sigma'}}}
\right),
\qquad \sigma,\sigma'\in \{+,-\}.
\end{align*}
Since, by hypothesis,
\[
u_{i+1}>u_i,\qquad \beta-u_{n} \geq \frac{\beta}{n+1}
\]
and $r(l)>s(l)$, we have that
\[
e^{-\left(\beta - (u_{r(l)}-u_{s(l)})\right)\omega_{\sigma'}} \leq 
e^{-\frac{n}{n+1}\beta \omega_{\sigma'}}.
\]
Hence
\[
\tilde{\Phi}(\mathbf{P})\doteq\tilde{\mathcal{Q}}_{-+}\tilde{\mathcal{Q}}_{--}\Phi(\mathbf{P})
\]
is rapidly decreasing, in every direction, because, $F_G$ is a microcausal functional and $\Phi(\mathbf{P})$ is the restriction on a particular subdomain of $\hat{\Psi}(-P,P)$ which is an entire analytic function which grows at most polynomially. Hence, the negative frequencies are exponentially suppressed and if directions containing only positive frequencies are considered, they are also rapidly decreasing by Proposition \ref{eq:decay-p}.
The integral over $\mathbf{P}$ can now be taken
and we
may apply Proposition \ref{pr:decay} to estimate the decay of the result of that integral. We obtain 
\[
|F_G(U,\mathbf{X})| \leq c' \prod_{l\in E(G)} e^{- M_- \sqrt{|\mathbf{x}_{r(l)}-\mathbf{x}_{s(l)}|^2}}\leq 
c' e^{- \frac{M_-}{\sqrt{n}} \sqrt{\sum_{i=1}^n |\mathbf{x}_{i}|^2}}
\]
where the constant $c'$ does not depend on $u_i$.
In the last inequality we used the fact that $G$ is a connected graph and thus every $x_{i}$ can be reached from the origin $(\mathbf{x}_0=0)$. Hence
\[
\sum_{l\in E(G)} \sqrt{{|\mathbf{x}_{r(l)}-\mathbf{x}_{s(l)}|^2}} \geq 
\text{max}_i{\sqrt{|\mathbf{x}_{i}|^2}} \geq \sqrt{\frac{1}{n}\sum_{i=1}^n |\mathbf{x}_{i}|^2}
\]
thus concluding the proof. 
\end{proof}

\begin{theorem} 
Let $A \in\mathcal{A}_I(\mathcal{O})$ where $\mathcal{O}\subset\Sigma_\epsilon$, 
the adiabatic limit 
\begin{gather*}
\omega^{\beta,V}(A)=
\lim_{h\to 1}\sum_{n}\int_{0\le u_1\le\dots u_n\le\beta} d u_1\dots d u_n\int_{\RR^{3n}}d^3\mathbf{x}_1\dots d^3\mathbf{x}_n 
h(\mathbf{x}_1)\dots h(\mathbf{x}_n)\nonumber\\
 \omega_T^\beta\!\left(A;\tau_{iu_{1},\mathbf{x}_{1}}(K);\dots;\tau_{iu_n;\mathbf{x}_{n}}(K)\right),
\end{gather*}
where $K\doteq\lim_{h\to 1} \mathcal H_I(0)$,
exists in the sense of perturbation theory and defines an equilibrium state for the interacting theory.
\end{theorem}
\begin{proof} Since $\mathcal{O}$ is of compact support, it exists and $R>0$ such that the open ball $B_R$ centered in the origin of Minkowski spacetime contains $\mathcal{O}$, namely $\mathcal{O}\subset B_R$.
Furthermore, thanks to the temporal cutoff $\chi$, and in view of the causal properties of the Bogoliubov map, $K = \lim_{h\to1} \mathcal {H}_I$ is supported in $B_R$ for a sufficiently large $R$.
 
Consider the $n-$th order contribution in the sum defining $\omega^{\beta,V}_h$ given in \eqref{eq:KMS-h-finite}
\begin{equation}
\label{eq:Omega-n}
\Omega_{n,h}(A)\doteq\int_{0\le u_1\le\dots u_n\le\beta} d u_1\dots d u_n\int_{\RR^{3n}}d^3\mathbf{x}_1\dots d^3\mathbf{x}_n 
h(\mathbf{x}_1)\dots h(\mathbf{x}_n) F(u_1,\mathbf{x}_1;\dots ;u_n,\mathbf{x}_n)
\end{equation}
where
\begin{equation*}
F(u_1,\mathbf{x}_1;\dots ;u_n,\mathbf{x}_n)\doteq \omega^\beta_T\!\left(A;\tau_{iu_{1},\mathbf{x}_{1}}(K);\dots;\tau_{iu_n;\mathbf{x}_{n}}(K))\right), \qquad A\in\mathcal{A}_I(\mathcal{O}).
\end{equation*}
To apply the results of Theorem \ref{th:cluster} we observe that, if $R$ is sufficiently large  $\mathcal H_I(0) \in \mathcal{A}_I(B_R)$, furthermore,
the form of the integration domain of the $u$ variables as given in \eqref{eq:Omega-n} is such that 
\begin{equation}\label{eq:int-domain}
0 \leq  u_1  \leq \dots \leq u_n \leq \beta. 
\end{equation}
Using the KMS condition we might restrict attention to the case where, $\beta- {u_{n}} \geq \frac{\beta}{n+1}$.
In fact, if this is not the case, there must exist an $m$ for which $u_{m}-u_{m-1} \geq \frac{\beta}{n+1}$.
Actually, for $A_i\in\mathcal{A}_I(\mathcal{O})$ by the KMS condition we have that
\begin{gather*}
\omega_T^\beta(\tau_{iu_0}(A_0) ; \tau_{iu_1,{\mathbf{x}}_1}(A_1); \dots ; \tau_{iu_n,{\mathbf{x}}_n}(A_n))\\
=\omega_T^{\beta}(\tau_{iu_{m},\mathbf{x}_{m} } (A_{m}) ; \dots ; \tau_{iu_n,{\mathbf{x}}_n}(A_n)
\otimes \tau_{i\beta+iu_0}(A_0) ; \dots ; \tau_{i\beta+iu_{m-1},{\mathbf{x}}_{m-1}}(A_{m-1}) )
\end{gather*}
hence we might now consider
\begin{gather*}
F'(v_1,\mathbf{y}_1;\dots ;v_n,\mathbf{y}_n)\doteq \\\omega^\beta_T(K ; \tau_{iv_1,\mathbf{y}_1}(K);\dots ; \tau_{iv_{n-m},\mathbf{y}_{n-m}}(K);\tau_{iv_{n-m+1},\mathbf{y}_{n-m+1}}(A_0);\tau_{iv_{n-m+2},\mathbf{y}_{n-m+2}}(K);\dots \tau_{iv_n,\mathbf{y}_v}(K))
\end{gather*}
in place of $F$.
In fact the previous equality obtained with the KMS condition together with translation invariance of the state implies that 
\[
F(u_1,\mathbf{x}_1;\dots ;u_n,\mathbf{x}_n) = F'(v_1,\mathbf{y}_1;\dots ;v_n,\mathbf{y}_n)
\]
if 
\[
(\mathbf{y}_1,\dots ,\mathbf{y}_n) = (\mathbf{x}_{m+1}-\mathbf{x}_m,\dots, \mathbf{x}_{n}-\mathbf{x}_m, -\mathbf{x}_m,\mathbf{x}_1-\mathbf{x}_m,\dots ,\mathbf{x}_{m-1}-\mathbf{x}_m)
\]
and
\[
(v_1,\dots ,v_n) = (u_{m+1}-u_m,\dots, u_{n}-u_m,\beta -u_m,\beta+u_1-u_m,\dots ,).
\]
The arguments of the function $F'$ have the desired property, actually $\beta - v_n= u_m - u_{m-1} \geq \beta /(n+1)$. We might thus use $F'$ in place of $F$,
because the integration over the $u$ variables is over a compact set and becasue the points where $u_i=u_j$ for some $i\neq j$ forms a zero measure set.
Hence, the Theorem \ref{th:cluster} implies that the integral over 
$\mathbf{x}_i$ can be taken for all $i$
to conclude the proof.
\end{proof}

\section{Spontaneous symmetry breaking and the Goldstone Theorem}

The model we are considering possesses an internal $U(1)$ symmetry. Actually, the Lagrangian is invariant under transformations  
\[
\varphi = U(\theta)\varphi \doteq e^{i\theta} \varphi 
\]
where $\theta \in [0,2\pi]$. 
However, in the state which describes the condensate, this symmetry is spontaneously broken. 
From the Goldstone theorem we expect that a massless (gapless) mode is present in the model. This observation is in contrast with the analysis discussed in the previous section. Actually there, all the fields in the linearized theory were assumed to be massive. Notice that, once the background $\phi$ is fixed in the decomposition $\varphi=e^{-i\mu x^0}(\phi+\psi)$, the Lagrangian for the linearized theory, is not invariant under $U(1)$ transformations 
\[
\psi \to e^{i\theta} \psi + (e^{i\theta}-1)\phi,
\] 
hence, the fact that both linearized fields $\psi_i$ are massive, 
is not in contrast with the Goldstone theorem.
Furthermore, if the Goldstone theorem holds for the full theory, this would imply that at least one gapless mode should exists if the full perturbation series is considered.

In the case of thermal theories, the proof of Goldstone theorem is not completely straightforward as for theories at zero temperature because the original proof makes use of Lorentz invariance \cite{GoldstoneSalamWeinberg}
(see also the work of Jona-Lasinio using effective action methods in \cite{JonaLasinio}). The equilibrium states are however not Lorentz invariant because of the presence of a preferred time direction in the KMS condition.
Furthermore, even if a gapless mode exists
 the particle content of the gapless mode is not immediately evident,
  as discussed by Bros and Buchholz in \cite{BrosBuchholz}.
 
The presence of Goldstone modes at finite temperature, has been discussed in \cite{Kowalski} using effective action methods. 
Based on the analysis of Swieca \cite{Swieca}, a proof of the Goldstone theorem without using Lorentz invariance has been given by Morchio and Strocchi in \cite{MorchioStrocchi}, see also the book \cite{Strocchi} for the application of similar ideas for the analysis of the case of finite temperature. 
Furthermore, the analysis of the slow decay of large spatially separated correlation functions in the presence of spontaneous symmetry breaking is discussed in \cite{JakelWreszinski}.
However, when a non trivial background is present as for the case of Bose-Einstein condensation, 
we don't expect that the presence of a gapless mode is directly related to the clustering properties of the correlation functions for large spatial separation. 
As an example, consider the two-point function of the state
$\omega^\pm_{\beta,c}$ discussed in section \ref{sec:free-theory} in the limit of vanishing temperature namely $\beta\to\infty$. The obtained state is the composition of the massive vacuum $\omega_0$ with the map $\gamma^\pm_c$ given in \eqref{eq:aut-gamma}.
Even if one of the modes in the two-point function of  $\omega^\pm_{\beta,c}$ is gapless, the clustering properties of 
$\omega^\pm_{\beta,c}$ 
 are equivalent to the one of the vacuum because $\gamma^\pm_c$ does not change the localization of the observables.  

The mentioned proofs cannot be directly applied for perturbatively constructed theories,
 for this reason in the next section we shall give a proof of the validity of Goldstone theorem 
  which can hold in our setting. 
We shall actually follow Swieca's proof without making use of Lorentz invariance. 

\subsection{Proof of the Goldstone theorem}

Here we would like to give a proof of the validity of Goldstone theorem at finite temperature in the presence of a condensate. For this purpose, we observe that the $U(1)$ invariance of the Lagrangian density $\mathcal{L}$ for $\varphi=\varphi_1+i\varphi_2$ is such that
\begin{equation}\label{eq:U1-symm}
(\varphi_1,\varphi_2)\to R(\theta)(\varphi_1,\varphi_2)
\end{equation}
where $R$ is a rotation of an angle $\theta$. Its infinitesimal version is
\[
\varphi_m \to\varphi_m+\epsilon t_{mn} \varphi_n  
\]
where $t$ is the antisymmetric metric
\[
t= \begin{pmatrix}0&-1\\1&0\end{pmatrix}.
\]
Associated to the symmetry which is spontaneoulsy broken in the state $\omega$ there is a current $J$ which is conserved.
This current is defined as
\begin{equation}\label{eq:J-definition}
J^\mu \doteq \frac{\delta \mathcal{L}}{\delta \partial_\mu\varphi_m }t_{mn}\varphi_n  =  i (   \overline{\varphi}\partial^\mu {\varphi}
-\partial^\mu \overline{\varphi}\varphi  ).
\end{equation}
By Noether theorem the action possesses the desired $U(1)-$symmetry if and only if the current $J$ is conserved, namely if $\nabla_\mu J^\mu =0$.
Following 
\cite{KRS}
we can now introduce a regularized charge operator associated to the current density $J^0$ introduced above.
Let $f\in C^{\infty}_0(\mathbb{R})$ be a time cutoff with $\supp{f}\in(-\epsilon,\epsilon)$, $f\geq 0$ and $\|f\|_1=1$. Furthermore, $g\in C^\infty_{0}(\mathbb{R}^3)$ is a space cutoff, $g(\mathbf{x})=1$ for $\mathbf{x}<1$.
The {\it regularized charge operator} associated to $J$ can be seen as the large $R$ limit of  
\begin{equation}\label{eq:QR}
Q_R \doteq \int  d^4x  f(x_0) g\left(\frac{\mathbf{x}}{R}\right)  J^0(x).
\end{equation}
The charge operator can be used to implement the infinitesimal $U(1)$ transformation of the field\footnote{
In perturbation theory equation \eqref{eq:commutator-current} can be proved starting from the master ward identity, 
\[
\partial^\mu_y T(J_\mu(y)\varphi_n(x)) = \delta(y-x) t_{nm}\varphi_m(x)
\]
using the current conservation and the causal properties of the commutator.
For further details in the case of QED we refer to \cite{DuetschFredenhagen}.
}
\begin{equation}\label{eq:commutator-current}
\lim_{R\to\infty}\left[Q_R,\varphi_m(t,\mathbf{y})\right] = t_{mn} \varphi_n(t,\mathbf{y}).
\end{equation}
Hence, in a state where the symmetry is spontaneously broken, namely when $\omega(\varphi_n(0))=\phi_n\neq 0$ for some $n$,
\begin{equation}\label{eq:constraint-goldstone}
\lim_{R\to\infty} \omega([ Q_R,\varphi_n(0)])  
=t_{nm} \omega({\varphi}_{m}(0)) 
=t_{nm}{{\phi}}_{m}.
\end{equation}
Notice that, in view of the support properties of $f$, $g$, 
and of the conservation of the current $J$, 
for $R$ sufficiently large we have that $\omega([ Q_R,\varphi_n(0)])$ is constant. Hence, the limit $R\to\infty$ can be safely taken and the final result does not depend on the particular form of $g$.

We are now ready to state the Goldstone theorem in the following form
\begin{theorem}\label{th:goldstone} (Goldstone)
Consider a complex scalar quantum field $\varphi$ whose Lagrangian density $\mathcal{L}$ possesses an $U(1)-$symmetry generated by the current $J$ given in \eqref{eq:J-definition}. 
Consider the distribution
\begin{align*}
G_n(x)\doteq\omega([J^0(x),\varphi_n(0)]).
\end{align*}
If the symmetry generated by $J$ is spontaneously broken in the state $\omega$, namely $\omega(\varphi_n(0))=\phi_n\neq 0$ for some $n$,
then in the spectrum of $G_n$ there is a zero frequency (gapless) contribution 
at vanishing momentum, 
namely 
$\hat{G}_n$, the Fourier transform of $G_n$, is such that
\begin{equation}\label{eq:goldstone-thesis}
\lim_{\mathbf{p}\to0} \hat{G}_n(p_0,\mathbf{p}) = \delta(p_0) t_{nm}\phi_m
\end{equation}
in the sense of distribution.
\end{theorem}
\begin{proof}
The support properties of the distribution $F^i_n(x) \doteq \omega([J^i(x),\varphi_n(0)])$ and invariance under spatial rotations imply that there exists a distribution $\mathcal{F}_n$ such that
\begin{align*}
\int  d^4x  f(x_0) g\left(\frac{\mathbf{x}}{R}\right) \omega([J^i(x),\varphi_n(0)])
&=\int d^4p
 \hat{f}(p_0) \hat{g}(\mathbf{p}R)R^3  \mathbf{p}^i \mathcal{F}_n(p^0,|\mathbf{p}|)\\
&=\int d^4p
 \hat{f}(p_0) \hat{g}(\mathbf{p})  \frac{\mathbf{p}^i}{R} \mathcal{F}_n\left(p^0,\frac{|\mathbf{p}|}{R}\right) ,
 \qquad i\in \{1,2,3\}
\end{align*}
where $f$ and $g$ are chosen as in the definition of regularized charge $Q_R$ given in \eqref{eq:QR}.
Causality implies that for $R$ sufficiently large, the left hand side does not depend on $R$.
Hence $\int dp_0
 \hat{f}(p_0) {\mathbf{p}^i} \mathcal{F}_n(p^0,{|\mathbf{p}|})$ must be bounded near $\mathbf{p}=0$ and as a consequence of this fact it holds that 
\begin{align*}
\lim_{R\to\infty}\int  d^4x  f(x_0) g\left(\frac{\mathbf{x}}{R}\right) \sum_{i=1}^3\nabla_i\omega([J^i(x),\varphi_n(0)]) = 
\lim_{R\to\infty}
\int d^4p
 \hat{f}(p_0) \hat{g}(\mathbf{p})  \frac{|\mathbf{p}|^2}{R^2} \mathcal{F}_n\left(p^0,\frac{|\mathbf{p}|}{R}\right) =0.
\end{align*}
Hence, current conservation furnishes a condition for $G_n$, namely 
\begin{align*}
\lim_{R\to\infty}\int  d^4x  f(x_0) g\left(\frac{\mathbf{x}}{R}\right) \sum_{\nu=0}^3\nabla_\nu\omega([J^\nu(x),\varphi_n(0)]) 
&= 
\lim_{R\to\infty}\int  d^4x  f(x_0) g\left(\frac{\mathbf{x}}{R}\right) \nabla_0\omega([J^0(x),\varphi_n(0)]) \\
&= 
\lim_{R\to\infty}\int  d^4p \hat{f}(p_0) \hat{g}(\mathbf{p}) p_0 \hat{G}_n\left(p_0,\frac{\mathbf{p}}{R}\right)  d^4p =0.
\end{align*}
Another constraint on the form of $G_n$ is given by
\eqref{eq:constraint-goldstone}, namely
\[
\lim_{R\to\infty}\int  d^4p \hat{f}(p_0) \hat{g}(\mathbf{p})  \hat{G}_n\left(p_0,\frac{\mathbf{p}}{R}\right)  d^4p   
=
t_{nm} \phi_m \hat{f}(0).
\]
Both conditions imply that \eqref{eq:goldstone-thesis} holds in the sense of distributions.
\end{proof}

As discussed in \cite{BrosBuchholz}, no direct particle interpretation can be inferred from \eqref{eq:goldstone-thesis}.
Actually,  
the singularity in $\hat{G}_n$ can be proven to exists only at $\mathbf{p}=0$ and not on the whole null cone. 
For a particle interpretation of this fact we refer to the paper \cite{BrosBuchholz}.

\subsection{Analysis of the validity of Goldstone theorem in perturbation theory}

We have seen that Goldstone theorem in the form stated in Theorem \ref{th:goldstone} holds for a quantum scalar field theory if the corresponding Lagrangian density
$\mathcal{L}$ is invariant under the $U(1)$ transformations given in \eqref{eq:U1-symm}, and if this symmetry is spontanously broken in a state $\omega$. 
Hence, in a generic quantum field theory, we can apply Goldstone theorem if the current $J$ given in \eqref{eq:J-definition} is conserved and if $\omega(\varphi_n(0))=\phi_n\neq 0$ for some $n$. 

\bigskip

We now check that at linear order some of the desired hypotheses are not satisfied.
In particular, we see that in the linearized theory 
the internal $U(1)$ symmetry is explicitly broken.
Actually, the current $J$ defined as in \eqref{eq:J-definition}, whose time component has the explicit form \eqref{eq:J-linear-theory}, is not conserved. To check this fact notice that
\begin{equation}\label{eq:conservation-J1}
\begin{aligned}
\partial^\mu J_\mu &= 2(\psi_2\square \psi_1-\psi_1\square \psi_2)
-2 \phi \square \psi_2 -4 \mu (\phi \dot\psi_1+\psi_1\dot\psi_1 +\psi_2\dot\psi_2)\\
&= 2\psi_1\psi_2(M_1^2-M_2^2) - 2\phi M_2^2 \psi_2
\end{aligned}
\end{equation}
where we have used the equation of motion given at linear order \eqref{eq:linearized-equations}. 
If now $\mu^2>m^2$ we have that $\phi\neq 0$ because $\lambda \phi^2 = (\mu^2-m^2)$. Furthermore, $M_1^2=(m^2-\mu^2) + 3\lambda \phi^2$, and $M_2^2=(m^2-\mu^2) + \lambda \phi^2$ and thus even if $M_2=0$, we have that $M_1^2-M_2^2\neq 0$ and hence $\partial^\mu J_\mu\neq 0$.

\bigskip

We pass now to analyse the interacting case. We observe that, if the full interacting equation of motion is used in evaluating $\partial^\mu J_\mu$, namely taking into account $\mathcal{L}_3$ and $\mathcal{L}_4$, we have that equation \eqref{eq:conservation-J1} needs to be changed to 
\begin{align*}
\partial^\mu J_\mu= 2\psi_1\psi_2(M_1^2-M_2^2-2\lambda \phi^2) - 2\phi M_2^2 \psi_2.
\end{align*}
Now both $M_2=0$ and $M_1^2-M_2^2-2\lambda \phi^2=0$ and hence the symmetry is not explicitly broken in the full classical theory. Notice that this analysis does not depend on the splitting between free and interacting Lagrangian, hence the eventual thermal mass contributions do not alter this analysis. 

\bigskip
We now check that conservation of the current $J$ holds in the case of interacting quantum field theory treated with  perturbation methods.
The very same analysis holds up  to the quotient with respect to the free equation of motion.
Actually, the Schwinger-Dyson equation implies that 
\[
R_V\left(\frac{\delta S}{\delta \varphi} \right)=
\frac{\delta S_0}{\delta \varphi} \neq 0
\]
where $S$ is the action constructed with the full Lagrangian density $\mathcal{L}$ and $S_0$ is the action of the linearized theory, constructed with the lagrangian density $\mathcal{L}_2$.
If expectation values on a quantum state are considered we have
\[
\omega\left (R_V\left (\frac{\delta S}{\delta \varphi} \right)\right) = 0.
\]
We now analyze in the same spirit the conservation of the current density $J$. To this end, we observe that
\begin{equation}\label{eq:Cons1}
R_V(\partial^\mu J_\mu) = 
i R_V(\overline{\varphi}\square\varphi - \square\overline{\varphi}\varphi) = 
R_V\left(\overline\varphi\frac{\delta S}{\delta \overline\varphi}-
\varphi\frac{\delta S}{\delta \varphi} \right)
=
R_V\left(\overline\varphi \cdot_T \frac{\delta S}{\delta \overline\varphi}-
\varphi \cdot_T  \frac{\delta S}{\delta \varphi} \right)
\end{equation}
where the last equality holds because of the properties of the time-ordered product and $\overline{S}=S$. We stress that the divergences present in 
$\overline\varphi(x)\cdot_T \delta S/\delta \overline\varphi(x)$ and proportional to $\lim_{y\to x}\Delta_F(x,y) \frac{\delta^2 S}{\delta\varphi(y)\delta\overline\varphi(x)}$ are  cancelled by the subtraction of $\varphi\cdot_T \delta S/\delta \varphi$.
To proceed we need the Master Ward Identity which can be shown to hold for this theory without anomalies.
Master Ward identities have been discussed in \cite{DF03, DF04, Hollands}.
For our purposes we start with equation (70) of \cite{FredenhagenRejznerBV}. Rewriting it in our context, we obtain the desired equation\footnote{A direct proof of \eqref{eq:MWI} can be obtained studying 
\[
\lim_{y\to x} R_V\left(\varphi(y)\cdot_T\frac{\delta S}{\delta \varphi}(x)\right) 
=
\lim_{y\to x} R_V\left(\varphi(y)\right)\star R_V\left(\frac{\delta S}{\delta \varphi}(x)\right) 
=
\lim_{y\to x} R_V\left(\varphi(y)\right)\star \frac{\delta S_0}{\delta \varphi}(x)+R_V(\varphi(y))\star D\varphi(x)
\]
where the limit is taken in a direction where $y$ is always in the future of $x$. The first equality is a consequence of the causal factorisation property and the last equality is the Schwinger-Dyson equation}
\begin{equation}\label{eq:MWI}
R_V\left(\overline\varphi\cdot_T \frac{\delta S}{\delta \overline\varphi}-
\varphi\cdot_T\frac{\delta S}{\delta \varphi} \right)
=
R_V(\overline\varphi) \star \frac{\delta S_0}{\delta \overline\varphi}-
R_V(\varphi)\star\frac{\delta S_0}{\delta \varphi}.
\end{equation}
Notice that, the divergences on the right hand side of the previous equality due to the pointwise multiplication with the quantum product vanishes because $\delta S_0/\delta \varphi$ is in the ideal of the linear equation of motion.
Hence, in any quantum state
\[
\omega(R_V(\partial^\mu J_\mu))
=
\omega\left(R_V(\overline\varphi) \star \frac{\delta S_0}{\delta \overline\varphi}-
R_V(\varphi)\star\frac{\delta S_0}{\delta \varphi}  \right) = 0
\]
We conclude that the current $J$ constructed with the interaction quantum scalar field is conserved in the sense of perturbation theory up to the ideal describing the equation of motion.
Finally we observe that even if this last observation has been given in terms of the field $\varphi$, since relations \eqref{eq:Cons1} and  \eqref{eq:MWI} are algebraic relations, the very same analysis holds for the fluctuations $\psi_1,\psi_2$.

We also observe that on the state we are considering 
$\omega(\varphi_1) \neq 0$ 
in the sense of perturbation theory. Actually $\phi = \frac{1}{\sqrt{\lambda}} \sqrt{\mu^2-m^2}$ and thus $\mathcal{L}_3$ is of order $\sqrt{\lambda}$ while $\mathcal{L}_4$ is of order $\lambda$ ($m_v^2$ needs to be chosen smaller than $\lambda c$). 
This implies that $\omega(R_V \psi_1)$ is at least of order $O(1)$ in $\lambda$ and hence it cannot totally cancel/compensate $\phi$. 

We thus have that the hypotheses of Theorem \ref{th:goldstone} holds in the sense of perturbation theory. Furthermore, we also notice that all the identities in the proof hold in the sense of perturbation theory. Hence we conclude that the thesis of that theorem holds in the sense of perturbation theory. 

\subsection*{Acknowledgments} 
NP thanks the ITP of the University of Leipzig for the kind hospitality during the preparation of part of this work and DAAD for supporting that visit with the program ``Research Stays for Academics 2017.'' It is a pleasure to thank the referees; their suggestions and requirements help in clarifying the content of our paper and stimulated, e.g. in the case of Gross-Pitaevskii limit, also further results.

\appendix

\section{The principle of perturbative agreement}\label{se:ppa}

Thanks to the {\it principle of perturbative agreement} \cite{HW05, DHP} the theories obtained with the two different splittings given in \eqref{eq:splitting1} and \eqref{eq:splitting2} of the Lagrangian density 
\[
\mathcal{L} = -\frac{1}{2}\partial \phi\partial \phi - \frac{\lambda}{4} \phi^4
\]
are equivalent.
In particular, the principle of perturbative agreement holds because it exists a map $\gamma$ which intertwines 
the time-ordered products constructed in theories with different masses $m$ and $m'$ which will be later set to 0. We shall here denote by $T_m$ the time ordering operator of a massive theory $\square -m^2$.
Hence, 
\[
\gamma: T_m \mathcal{F}_{\text{mloc}} \to  T_{m'} \mathcal{F}_{\text{mloc}}
\]
where here $\mathcal{F}_{\text{mloc}}$ is the commutative algebra 
of multilocal functionals (with respect to the pointwise (classical) product).
The non trivial nature of $\gamma$ can be seen from the fact that the time-ordered products $T_m$ and $T_{m'}$ are different.
At the same time, the perturbative agreement requires that the common part of $\mathcal{L}_I$ and $\mathcal{L}'_I$, or more generally all Wick powers, are left invariant up to ordinary renormalization freedom by the map $\gamma$. 
We shall here see that this is the case provided the time ordering operator $T_m$ are constructed with the preferred symmetrized Hadamard singularity $H_m$ of $\square -m^2$ as in \eqref{eq:H} with $W=0$ for some fixed lengthscale $\xi$, namely $T_m(F)=\no{F}_{H_m}$.  

For completeness we discuss this invariance under $\gamma$ in the simple case of the Wick square 
\[
\Phi^2(f) = \int \phi^2(x) f(x) d^4x.
\]
We want to compare $T_{m'} \Phi^2(f)$ with $T_m \Phi^2(f)$.
To this end we observe that 
\[
T_{m'}\Phi^2(x) = \lim_{y\to x}T_{m'}(\ph(x)\ph(y)) - \Delta_{F,{m'}}(x,y),
\]
and we study
\[
\Delta\Phi^2=\gamma  T_{m'}\Phi^2 - T_{m}\Phi^2.
\]
Suppose that, on regular functionals,  
$T_{m'}$ and $T_{m}$ are constructed starting with the Feynman propagators $\Delta_{F,{m'}}(x,y)$ and $\Delta_{F,{m}}(x,y)$, hence we have
\begin{align*}
\Delta\Phi^2(x)&=\lim_{y\to x}\left(\gamma T_{m'}(\ph(x)\ph(y)) -T_{m}(\ph(x)\ph(y)) - \Delta_{F,{m'}}(x,y) + \Delta_{F,{m}}(x,y)\right). 
\end{align*}
Furthermore, the principle of perturbative agreement (see e.g.\cite{DHP} Theorem 3.2 where $\gamma$ is denoted by $\beta$) implies that $\gamma \ph = \ph$ and that 
\[
\gamma  T_{m'}(\ph(x)\ph(y))    =  T_{m}(\ph(x)\ph(y))
\] 
and thus 
\begin{align*}
\Delta\Phi^2(x)
&=\lim_{y\to x}\left(- \Delta_{F,{m'}}(x,y) + \Delta_{F,{m}}(x,y)\right).
\end{align*}
It remains to compare $\Delta_{F,{m'}}$ with $\Delta_{F,m}$. Considering the form of the Hadamard singularities, we observe that 
\[
\lim_{y\to x}\left(\Delta_{F,{m'}}(x,y) -  \Delta_{F,{m}}(x,y) \right) =  \lim_{y\to x} ( V_{m'} - V_{m})\log \left(\frac{\sigma_\epsilon}{\mu^2}\right) + W(x,y)
\]
where $V_m$ is the standard Hadamard $V$ coefficient, $\sigma_\epsilon$ is the regularized squared geodesic distance, $\mu$ is a length scale and $W$ is a smooth reminder. In the coinciding point limit $( V_{m'} - V_{m})$ is proportional to $\delta m^2 = {m'}^2-{m}^2$, hence the limits present in the previous equation are divergent. However, the kind of divergences present in that coinciding point limit are not different than the standard divergences present in $R_Q(\ph^2)$ with 
\[
 Q \doteq \int g \delta m^2  \ph^2  d^4x, 
\]
if the time-ordered product of 
$T_m(Q\ph^2)$ is constructed without making use of a correct renormalization prescription.
We may thus avoid (renormalize) them in a similar way as the divergences present in the naive construction of $T_m$ are resolved. 
To this end we notice that, if we consider the Feynman operator as operator on functions, we have that 
\[
\Delta_{F,m'} = \lim_{g\to 1}\sum_{n\geq 0}\Delta_{F,m} \left( -g \delta m^2 \Delta_{F,m}\right)^n.
\]
We observe, by power counting, that the divergent contributions in the sum are the contributions $n=0$ and $n=1$. The contributions $n=0$ is removed by the subtraction of $\Delta_{F,m}$ while the divergences present in the contribution $n=1$ are similar to the divergences present in $T_m(QQ)$ and thus they can be treated with renormalization theory.
Actually, at order $1$ in $\delta m^2$ the remaining contribution  $\Delta\Phi^2_{(1)}$ in the difference 
$\gamma T_{m'}\Phi^2-T_{m}\Phi^2$ is thus
\[
\Delta\Phi^2_{(1)} = -\lim_{g\to 1}\int (\Delta_{F,m}^2(x-y) )_{\text{ren}} \delta m^2 g(y)d^4y.
\]
We regularize the product of Feynman propagators in the following way
\begin{equation}\label{eq:ren-freedom}
(\Delta_{F,m}^2)_{\text{ren}}(x) = (\square  + a )\int_{(2m)^2}^\infty dM^2 \frac{\rho_2}{M^2+a} i\Delta_{F,M}(x), \qquad \rho_2=\frac{1}{16\pi^2} \sqrt{1-\frac{4m^2}{M^2}}
\end{equation}
where $a$ is a parameter which takes into account the known renormalization freedom. If we fix it to 0 we have that 
\begin{align*}
-\Delta\Phi^2_{(1)} 
&= \lim_{g\to1}\frac{\delta m^2}{(2\pi)^{-4}}\int d^4p \;\hat{g}(p) (\hat{\Delta}_{F,m}^2)_{\text{ren}}(p) \\
&=\lim_{g\to1} \frac{\delta m^2}{(2\pi)^{-4}}
\int d^4p \;\hat{g}(p)
 (-p^2)\int_{(2m)^2}^\infty dM^2 \frac{\rho_2}{M^2} \frac{1}{-p^2+M^2+i\epsilon} =0 
\end{align*}
where in the last step we used the fact that  $\hat{g}$ tends to $\delta(p)$ in the limit $g\to 1$. 
Higher order contributions can be computed directly and they give a finite result.

\bigskip
Instead of performing these explicit computations, fur our purposes,  we need to understand the origin of these contributions with the following observation.
Actually, the essential steps in the previous computation are the following: we have computed the expansion in powers of $\delta m^2$ of $\Delta_{F,m'}(x-y)$, we have removed the contributions of order $0$ and $1$ and we have eventually taken the coinciding point limits. In formulae
\begin{equation}\label{eq:compact}
-\Delta \Phi^2 = \lim_{y\to x}    \left( \Delta_{F,m'}(x,y)-\Delta_{F,m}(x,y)-
\left.\frac{\partial}{\partial \delta m^2}\Delta_{F,m'}(x,y)\right|_{\delta m^2=0}\delta m^2 \right).
\end{equation}
We prove now that if one starts with the time-ordered propagators constructed with the preferred Hadamard functions $H$ instead of the Minkowski vacuum $\Delta \Phi^2$ vanishes, we have actually that 
\[
H_{F,m'} = \frac{U}{\sigma_\epsilon} + V\log\left(\frac{\sigma_\epsilon}{\mu^2}\right)
\]
and expanding $V$ in powers of $\sigma$ we have
\[
V= c  \frac{{m'}^2}{m'\sqrt{\sigma}}I_1(m'\sqrt{\sigma}) =  \frac{c}{2}\left( m'^2    + \frac{1}{8}{m'}^4 \sigma +  \dots  \right).
\]
Hence, 
we conclude that 
\[
-\Delta \Phi^2 = \lim_{y\to x}    \left( H_{F,m'}(x,y)-H_{F,m}(x,y)-
\left.\frac{\partial}{\partial \delta m^2}H_{F,m'}(x,y)\right|_{\delta m^2=0} \delta m^2\right) = 0.
\] 
The same argument essentially holds also on any curved background, so in the case of mass renormalization we see that there is a choice of renormalization freedom ($a=0$ in \eqref{eq:ren-freedom}) such that
\[
\gamma T_{m'}\Phi^2 =T_m\Phi^2.
\] 
where now the Wick powers $T_m\Phi^2 = \no{\Phi^2}_{H_m}$ are constructed (regularized) with respect to the distinguished Hadamard singularity $H_m$ and not with respect to the symmetric part the two-point function of a state.
Finally, we observe that, the same results holds also when Wick monomial of higher order are considered, namely $\gamma T_{m'}\Phi^n=T_{m}\Phi^n$.

\section{Nonrelativistic limit of the free complex scalar field}\label{sec:non-rel-limit}
In section \ref{sec:free-theory} we have analyzed the possible KMS states with $\be>0$ and chemical potential $\mu$.
In particular, for $\beta>0$ and chemical potential $|\mu|<m$ we have $\omega_{\beta,\mu}$, 
for $\beta>0$ and chemical potential $\mu=\pm$ and a condensate $c$  we have $\omega_{\beta,c}^{\pm}$.
We discuss in this appendix the nonrelativistic limit of these states. 
We shall furthermore see that the charge density converges to the particle density under that limit.

The nonrelativistic limit is obtained for temperatures $T=\be^{-1}$ which are small compared to $m$ and for velocities $v\ll1$ such that $mv^2/T=O(1)$. The chemical potential is near to $m$, $(m-\mu)/T=O(1)$.
We set $\bar\mu=\la^{-2}(m-\mu)$ and 
\begin{equation}
\psi_{\la}(t,\mathbf{x})=\sqrt{2m\la^3}\ \ph(\la^2 t,\la \mathbf{x})e^{-it\la^2m}
\end{equation}
and find for the first class of KMS states

\begin{equation}
\omega_{\la^2\beta,m-\la^{-2}\bar{\mu}}(\psi_{\la}^*(t,x)\psi_{\la}(t',x'))=\frac{1}{(2\pi)^3}\int d^4 p\frac{\delta(p^2+m^2)\epsilon(p_0)e^{i(p_0-m)\la^2(t-t')-i\mathbf{p}\la(\mathbf{x}-\mathbf{x'})}}{1-e^{-\la^2\be(p_0-m+\la^{-2}\bar{\mu})}}
\end{equation}
\begin{equation}
\stackrel{\la\to\infty}{\rightarrow}\frac{1}{(2\pi)^3}\int d^3\mathbf{p}\frac{e^{i\frac{\mathbf{p}^2}{2m}(t-t')-i\mathbf{p}(\mathbf{x}-\mathbf{x}')}}{1-e^{-\be(\frac{\mathbf{p}^2}{2m}+\bar{\mu})}}=\omega_{\beta,\bar{\mu}}(\psi^*(t,\mathbf{x})\psi(t',\mathbf{x}'))
\end{equation}
which is the 2-point function of the $\be$-KMS state with chemical potential $\bar{\mu}$ of the nonrelativistic free scalar field $\psi$.
For the second class (with condensate) we set $\bar{\mu}=0$ and consider a sequence of condensates
\begin{equation}
c_{\lambda}(\mathbf{x})=(2m\la^3)^{-\frac12}c(\la^{-1}\mathbf{x})
\end{equation}
where $c$ is a harmonic function. Then the states $\omega^{+}_{\la^2\be,c_{\la}}$ converge to the state of $\psi$ with the 2-point function
\begin{equation}
\omega_{\beta,c}(\psi^*(t,\mathbf{x})\psi(t',\mathbf{x}'))=\frac{1}{(2\pi)^3}\int d^3\mathbf{p}\frac{e^{i(\frac{\mathbf{p}^2}{2m})(t-t')-i\mathbf{p}(\mathbf{x}-\mathbf{x}')}}{1-e^{-\be\frac{\mathbf{p}^2}{2m}}}+\overline{c(\mathbf{x})}c(\mathbf{x}').
\end{equation}
Note that the contributions of antiparticles disappears in both cases in the limit $\lambda\to\infty$ due to the fact that the chemical potential $\mu$ tends to $+m$. Replacing $+m$ by $-m$ exchanges the role of particles and antiparticles.
Note furthermore that the hermitian scalar field does not have a meaningful nonrelativistic limit.
Actually, one sees
that the corresponding quantum states are not stable against local
perturbations. This happens because local perturbations do not com-
mute with particle number.
We now pass to analyze the charge density in the nonrelativistic limit
Let $J_0d^3\xx=-i\no{\ph^*\dot{\ph}-\ph\dot{\ph}^*}_Hd^3\xx$ denote the charge density of the complex scalar field. We scale $\xx$, $\be$ and $m-\mu$ as before and obtain
\begin{equation}
\lim_{\la\to\infty}\omega_{\la^2\be,m-\la^{-2}\bar{\mu}}(J_0)d^3(\la\xx)=\left(\int d^3\pp \frac{1}{e^{\beta(\frac{\pp^2}{2m}+\bar{\mu})}-1}\right)d^3\xx
\end{equation}
for $\bar{\mu}>0$ and
\begin{equation}
\lim_{\la\to\infty}\omega_{\la^2\be,c_{\lambda}}(J_0)d^3(\la\xx)=\left(\int d^3\pp \frac{1}{e^{\beta\frac{\pp^2}{2m}}-1}+|c(\xx)|^2\right)d^3\xx
\end{equation}
\ie the charge density tends to the particle density in the nonrelativistic limit.

\section{First Hadamard Coefficients for $D$}\label{sec:hadamard-coefficient}

We compute in this section the first Hadamard coefficients for the operator $D$ given in \eqref{eq:D-form}. 
We recall that in this case the Hadamard singularity has the following structure
\begin{equation}\label{eq:H-D}
H +\frac{i}{2}\Delta= \lim_{\epsilon\to0^+}\frac{U}{\sigma_\epsilon} + V \log \left(\frac{ \sigma_\epsilon}{\xi^2}\right), \qquad V = \sum_n V_n\sigma^n
\end{equation}
similar to \eqref{eq:H} with vanishing $W$,
where now $U,V,V_i$ are $2\times 2$ matrices and where $\sigma_\epsilon$ is again one half of the regularized geodesic distance. In the case of Minkowski $\sigma(x,y)=\frac{1}{2}(x-y)^\mu(x-y)_\mu$.
The requirement that $D H$ is smooth implies the following transport equations
\begin{gather*}
 2\nabla^\mu \sigma \nabla_\mu U + U(\square \sigma -4)  -2i\mu \sigma_2 \;U \partial_0 \sigma = 0\\
- 2\nabla^\mu \sigma \nabla_\mu V_0 - V_0(\square \sigma -2)  +2i\mu \sigma_2 \;V_0\partial_0 \sigma +DU= 0\\
 DV=0.
\end{gather*}
The first two equations give 
\[
2\nabla^\mu \sigma \nabla_\mu (U^{-1}V_0) = - 2 U^{-1}V_0 - U^{-1}DU
\]
considering integrals along the geodesic $\gamma$ joining $x,y$ and indicating by $r$ its affine parameter ($\gamma(0)=x$, $\gamma(1)=y$), that equation gives 
\begin{align*}
2r\frac{d}{d r} (U^{-1}V_0) + 2U^{-1}V_0  = - U^{-1}DU \\
2\frac{d}{d r}  (r U^{-1}V_0)   = -  U^{-1}DU
\end{align*}
Integrating along $\gamma$ we get
\begin{equation}\label{eq:V0}
V_0(x,y)  =  -\frac{1}{2} U(x,y) \int_0^1 dr (U^{-1} DU) 
\end{equation}

\bigskip
The first transport equation can be solved once the initial condition $U(x,x)=\mathbb{I}$ is fixed (as we obtained in the previously computed vacuum state). 
To find its solution, due to translation invariance, it is enough to study
\[
2x^\mu \frac{\partial}{ \partial x^\mu} U(0,x) +2 i\mu x^0 \bm\sigma_2 U(0,x) = 0 
\]
a solution which satisfy the desired initial condition is 
\begin{equation}\label{eq:U}
U(0,x)= \cos(\mu x^0)\mathbb{I} - i\bm\sigma_2 \sin(\mu x^0) = \begin{pmatrix}
\cos(\mu x^0) & -\sin(\mu x^0) \\
\sin(\mu x^0) & \cos(\mu x^0)
\end{pmatrix}
\end{equation}
hence, 
since $U^{-1}= \cos(\mu x^0)\mathbb{I} + i\bm\sigma_2 \sin(\mu x^0)$
\begin{align*}
U^{-1}DU&= U^{-1}\left((-\square +M^2)\mathbb{I} + \delta M^2 \bm\sigma_3 + 2i\mu \bm\sigma_2 \partial_0 \right)U\\
&=(M^2+\mu^2) \mathbb{I} +\delta M^2 \left(\cos(2\mu x^0)\bm\sigma_3 -\sin(2\mu x^0) \bm\sigma_1\right)
\end{align*}
recalling \eqref{eq:V0} 
we have
\begin{align}
V_0(0,x) &= -\frac{1}{2} U(0,x) \int_0^1 dr\;  U^{-1} DU
\notag
\\
\label{eq:V0-1} &= -\frac{1}{2} U(0,x) \left(   (M^2+\mu^2) \mathbb{I} +\delta M^2
\left(\frac{\sin(2\mu x^0)}{2\mu x^0}\bm\sigma_3 
+\frac{\cos(2\mu x^0)-1}{2\mu x^0} \bm\sigma_1\right)
\right)  
\end{align}

We can now expand $V$ as $\sum_{n\geq 0} V_n \sigma^n$. The equation $DV=0$ and the knowledge of $V_0$, permits to compute $V_n$ recursively. We are in particular interested in $[V_1](x) = V(x,x)$ because this coefficient is proportional to the trace anomaly of the stress tensor of the linearized theory \cite{Wald,Moretti, HW05} and in particular enters in the expressions
\[
\psi_i D\psi_j, \qquad \partial_a \psi_i D \psi_j.
\] 
We observe that
\[
[V_1] = \frac{1}{4} [DV_0]
\]
furthermore we can expand $V_0=UX$ for some $X$
hence we have 
\[
[V_1] = \frac{1}{4}\left([DU][X]-[U][\square X]-2 [\nabla_\mu U][\nabla^\mu X]+2i\mu \sigma_2[\partial_0 X] \right)
\]
where
\begin{align*}
[U]&=\mathbb{I}
&[DU] &= (M^2+\mu^2) \mathbb{I} + \delta M^2 \bm\sigma_3 
&[X] &= -\frac{1}{2}[DU] \\
[\partial_0 X] &= -\frac{1}{2} \mu \delta M^2 \bm\sigma_1 
&[\partial_0 U] &= -i \mu \bm\sigma_2 
&[\square X] &= -\frac{2}{3} \mu^2 \bm\sigma_3 \delta M^2.
\end{align*}
Summarizing this analysis
\[
[V_1]=\frac{1}{4}[DV_0] = -\frac{1}{8} ((M^2+\mu^2)^2+\delta M^4) \mathbb{I} -\frac{1}{4} \left(M^2 + \frac{\mu^2}{3}\right) \delta M^2\bm\sigma_3
\]
We observe that $[V_1]$ is diagonal and constant, hence, we see that this anomaly is not visible in the conservation of the charge $J^\mu$ given in \eqref{eq:conservation-J1}.

\section{Some technical propositions}

We report here a technical proposition, similar to Proposition 9 in \cite{FredenhagenLindner}, and adapted to the case studied here.  
\begin{proposition}\label{eq:decay-p}
Consider $A_0,\dots, A_n$ in $\mathcal{A}_I$ and, for $k\in\mathbb{N}$ the following compactly supported distribution
\[
\Psi(z_1,\dots, z_k,y_1,\dots, y_k) \doteq \left.\left(\prod_{l=1}^k \frac{\delta}{\delta \psi^{s(l)}(z_l)}\otimes \frac{\delta}{\delta \psi^{r(l)}(y_l)} \right)
(A_0 \otimes A_1\otimes\dots A_n)\right|_{(\psi^0,\dots,\psi^n) = 0}
\]
where $s$ and $r$ maps $\{1,\dots, k\}$ to $\{0,\dots, n\}$ with the condition $s(l)<r(l)$.
The function 
\[
 \varphi:(p_1,\dots, p_k)\mapsto \hat{\Psi}(-p_1,\dots, -p_k,p_1,\dots, p_k)
\] 
 given in terms of  the Fourier transform $\hat{\Psi}$ of $\Psi$ is of rapid decrease if $P\in V_+^{k} \cup V_-^{k}$. Where $V_\pm$ denotes the forward/backward light cone in the cotangent space.
\end{proposition}
\begin{proof}
Since $A_i\in \mathcal{A}_I$, is a microcausal functional 
we have that 
\[
WF(\Psi(Y,Z)) \cap (\bigcap_{s(l)=i } W^+_l) \cap (\bigcap_{r(l)=i } W^+_{k+l}) = \emptyset, \qquad \forall i,
\] 
and the same holds with $W^-$ at the place of $W^+$,
where
\[
W^\pm_j \doteq (T^*\mathbb{M})^{\otimes j-1} \otimes \overline{V}^\pm  \otimes (T^*\mathbb{M})^{\otimes 2k-j-1}
\]
for $j\in (0,\dots, 2 k)$. 
Thanks to this property, we can prove that, if every component of $P=(p_1,\dots, p_k)$ is a future pointing causal vector and if $p_l = 0$ for all $l$ such that ${r(l)=i}$ for some $i$, 
$\varphi(P)$ can be of non rapid decrease 
only if
\[
\sum_{{s(l)=i}} p_l = 0. 
\]
With this observation we can prove by induction on $i$ that  $p_l=0$ for all $l$, actually, if $i=0$ there are no $l$ such that $r(l)=0$ and thus the previous condition implies that
\[
\sum_{s(l)=0} p_l = 0
\] 
and hence, since for all $l$ $p_l\in V_+$, for every $l$ such that $s(l)=0$ $p_l=0$. Furthermore, 
If we have already proved that $p_l=0$ for every $s(l)=j< i$, we can prove it also for $s(l)=i$. Actually, in that case we already know that for every $l$ such that $r(l)=i$, $s(l)<i$ and thus $p_l=0$. 
Hence, the direction $P$ we are analysing can be of non rapid decrease only if 
\[
\sum_{s(l)=i} p_l = 0
\]
which implies that for every $l$ such that $s(l)=i$, $p_l=0$. 
We have thus proved that all these directions are of rapid decrease because the direction $P$ can be of non rapid decrease only if $P=\{0\}$ and the zero section does not intersect $WF(\Psi)$.
\end{proof}

In the following proposition we prove the exponential decay of the two-point function of the KMS states for the linearized theory studied in section \ref{se:kms-linear}. This proposition is used in the proof of the clustering estimate necessary for the adiabatic limit.

\begin{proposition}\label{pr:decay}
Let $f\in \mathcal{E}'(M)$, with $\text{supp} f \subset C_R$ where $C_R$ is a sphere of radius $R$.  Consider
\[
I_\sigma(u,\mathbf{x}) =   
\frac{1}{(2\pi)^3} \int d^3\mathbf{p}  
\frac{ e^{i\mathbf{p}\mathbf{x}}e^{ -  \omega_\sigma u}}{(\omega_+^2-\omega_-^2)\omega_{\sigma}}
\hat{\overline{D}}(\pm\omega_\sigma,\mathbf{p})
\hat{f}(\omega_\sigma,\mathbf{p}), \qquad \sigma\in\{+,-\}
\]
where $\hat{\overline{D}}$ is as in \eqref{eq:Dhat}, see also \eqref{eq:D-form},  
 and $\omega_\pm$ as in \eqref{eq:omegapm}. Assume $M_->0$,
it holds that
\[
|I_\sigma(u,\mathbf{x})|  \leq c e^{-M_-r}, \qquad r=\sqrt{|\mathbf{x}|^2+u^2},
\qquad 
 r>>R, \qquad u>0
\]
\end{proposition}
\begin{proof}
We observe that
\begin{align*}
-\frac{\hat{\overline{D}}(\pm\omega_\sigma,\mathbf{p})}{(\omega_+^2-\omega_-^2)\omega_{\sigma}} 
&=
\left(   
\frac{\sigma 1}
{ 2 \omega_{\sigma}}
+
\frac{2\mu^2}
{(\omega_+^2-\omega_-^2)\omega_{\sigma}}
\right) \mathbb{I} + \frac{\delta M^2}
{(\omega_+^2-\omega_-^2)\omega_\sigma} \sigma_3 \pm \frac{2\mu}
{\omega_+^2-\omega_-^2} \sigma_2 
\end{align*}
Hence, we analyze separately the following functions
\begin{align*}
I^1_\sigma(u,\mathbf{x}) &=   
\frac{1}{(2\pi)^3} \int d^3\mathbf{p}  
\frac{1}
{2\omega_{\sigma}}
e^{i\mathbf{p}\mathbf{x}}e^{ -  \omega_\sigma u}
\hat{f}(\omega_\sigma,\mathbf{p}), \qquad \sigma\in\{+,-\}
\\
I^2_\sigma(u,\mathbf{x}) &=   
\frac{1}{(2\pi)^3} \int d^3\mathbf{p}  
\frac{1}
{2(\omega_+^2-\omega_-^2)\omega_\sigma}
e^{i\mathbf{p}\mathbf{x}}e^{ -  \omega_\sigma u}
\hat{f}(\omega_\sigma,\mathbf{p}), \qquad \sigma\in\{+,-\}
\\
I^3_\sigma(u,\mathbf{x}) &=   
\frac{1}{(2\pi)^3} \int d^3\mathbf{p}  
\frac{1}
{\omega_+^2-\omega_-^2}
e^{i\mathbf{p}\mathbf{x}}e^{ -  \omega_\sigma u}
\hat{f}(\omega_\sigma,\mathbf{p}), \qquad \sigma\in\{+,-\}
\end{align*}
$I$ is then a linear combination of $I^i$ with constant coefficients. We prove with some details the decay of $I^2_\sigma$ for $\sigma\in \{+,-\}$, analogous results holds also for the other components. 

Without loosing generality, let us assume that $\mathbf{x} = r\mathbf{n}$ with $n=(1,0,0)$. We shall discuss the decay for large $r$.
Notice that
\begin{align*}
I^2_\sigma(u,r) &=   
\frac{1}{(2\pi)^3} \int d^3\mathbf{p}  
\frac{1}
{(\omega_+^2-\omega_-^2)2\omega_\sigma}
e^{i p_1 r}e^{ -  \omega_\sigma u}
\hat{f}(\omega_\sigma,\mathbf{p})
\end{align*}
We evaluate the integral in $p_1$ with the help of complex analysis considering $p_1=z=x+iy$ with $x,y\in\mathbb{R}$ a complex variable.
We shall furthermore, consider a particular contour of integration  in the upper half plane.
The integral over one branch of the contour we shall chose corresponds to $I^2_\sigma$, furthermore the contour will be extended to infinity and chosen in such a way to avoid the poles in $1/(\omega_+^2-\omega_-^2)2\omega_\sigma$ and the brunch cuts present in $\omega_\sigma$. 

Actually, we observe that since $f$ is a compactly supported distribution, and its support is contained in the disc centered in $0$ and of radius $R$, 
by the Paley Wiener theorem, its Fourier transform is an entire analytic function. 
Furthermore, it grows at most polinomially in every real direction and exponentially in complex directions. Hence, it exist two constants $C>0$, $C'>0$ and a $N>0$ such that 
\begin{equation}\label{eq:paley-wiener}
|\hat{f}(p_0,\mathbf{p})| \leq C e^{R {|\text{Im}(p_0)|+|\text{Im}(\mathbf{p})|}} (1+|\text{Re}(p_0)|+|\text{Re}(\mathbf{p})|)^N  
\leq 
C' e^{R \sqrt{ |p_0|^2+|\mathbf{p}|^2}} 
\end{equation}
hence, if it is composed with $\omega_\sigma$ and if it is seen as a function of $p_1$, it is analytic everywhere up the branch cuts which are present in the principal squares of 
\begin{align}
\omega_\pm 
  &= \sqrt{w^2 +2\mu^2 \pm \sqrt{4\mu^4 + 4\mu^2w^2  +(\delta M^2)^2}}.
\end{align}
This implies that the integrand of $I^2_\sigma$, seen as a function of $p_1$ is analytic every where up to the poles of $1/(\omega_+^2-\omega_-^2)2\omega_\sigma$ and the branch cuts mentioned above. 
To describe the contour of integration we analyze the location of the branch cuts and the poles.

\bigskip
We study the form of the branch cuts of $\omega_-(z,\mathbf{p}_\perp)$.
In view of the definition of 
\[
\omega_-^2 = w^2 + 2\mu^2 - 2\mu \sqrt{w^2 +\mu^2 + \frac{\delta M^4}{4\mu^2}}
\]
and recalling that $w^2 = (z^2)+|\mathbf{p}_\perp|^2 + M^2$ where $z=x+iy$ is the complex variables which replaces $p_1$ and $\mathbf{p}_\perp=(p_1,p_2)$, we have that there is a branch cut where $Z(z)\doteq w^2 +\mu^2 + \frac{\delta M^4}{4\mu^2}\leq 0$ and where $W(z)\doteq\omega_-^2\leq 0$. 

The first condition, in the upper half complex plane, is met if 
\[
x=0, \qquad y\geq \sqrt{|\mathbf{p}_\perp|^2 + M^2+\mu^2 + \frac{\delta M^4}{4\mu^2}}.
\]
Writing $W(z)=c+id - 2\mu \sqrt{a+ib}$, with $a,b,c,d\in\mathbb{R}$, 
the second condition is met if
\begin{equation}
\begin{cases}
d -2\mu\text{Im}\sqrt{Z}= d - \sqrt{2}\mu \frac{b}{\sqrt{\sqrt{a^2+b^2}+a}} =0  \\
c -2\mu\text{Re}\sqrt{Z} = c -\sqrt{2}\mu \sqrt{\sqrt{a^2+b^2}+a} \leq 0   \label{eq:ineq}
\end{cases}
\end{equation}
In our case 
\begin{align}
a&=x^2-y^2 + |\mathbf{p}_\perp|^2 + M^2+\mu^2 + \left(\frac{\delta M^2}{2\mu}\right)^2
\notag\\
b&= d = 2 y x
\notag\\
c&= x^2 - y^2+ |\mathbf{p}_\perp|^2+ M^2+2\mu^2 
\label{eq:a-b}
\end{align}
hence, since $M_->0$, $M^2 > \delta M^2 \geq   0$, \eqref{eq:ineq} has solutions if 
\[
x=0, \qquad        |\mathbf{p}_\perp|^2 + M^2 -  \delta M^2 \leq       y^2 \leq |\mathbf{p}_\perp|^2 + M^2+  \delta M^2 
\]
or if 
\[
2\mu^2-a = \sqrt{a^2+b^2} 
\]
which holds if 
\begin{equation}\label{eq:cuts}
y = \sqrt{\frac{\mu^2 B^2}{\mu^2 -x^2}-\mu^2}
\end{equation}
where $B^2 = |\mathbf{p}_\perp|^2 + M^2+\mu^2 + \left(\frac{\delta M^2}{2\mu}\right)^2$. Hence \eqref{eq:cuts} has a solution for $|x|\leq \mu$ and it holds that $y \geq \sqrt{B^2-\mu^2}$. 
Furthermore, the inequality in \eqref{eq:ineq} gives
\[
c-2\mu^2 \frac{b}{d} = c-2\mu^2 \leq 0 
\]
which gives $y^2 \geq x^2 + |\mathbf{p}_\perp|^2 + M^2 $, which is always true. 

Summarizing, the cuts are the following curves in upper half complex plane 
\[
\gamma_1=\begin{cases}
x=0\\ 
y\geq \sqrt{|\mathbf{p}_\perp|^2 + M^2+\mu^2 + \frac{\delta M^4}{4\mu^2}}
\end{cases}
\qquad 
\gamma_2=\begin{cases}
x=0\\ 
\sqrt{  |\mathbf{p}_\perp|^2 + M^2 -  \delta M^2} \leq       y \leq \sqrt{|\mathbf{p}_\perp|^2 + M^2+  \delta M^2 }
\end{cases}
\]
and
\[
\gamma_3=\begin{cases}
|x|\leq \mu\\ 
y = \sqrt{\frac{\mu^2 \left(|\mathbf{p}_\perp|^2 + M^2+\mu^2 + \left(\frac{\delta M^2}{2\mu}\right)^2\right)}{\mu^2 -x^2}-\mu^2} \geq 
\sqrt{|\mathbf{p}_\perp|^2 + M^2 + \left(\frac{\delta M^2}{2\mu}\right)^2}.
\end{cases}
\]
As also displayed in the figure \ref{fig:cuts} we notice that $\gamma_1$ is contained inside of $\gamma_3$.
Furthermore, 
 $\gamma_2$ is not entirely contained in $\gamma_3$. On $\gamma_2$ and $\gamma_3$ $\omega_-$ is purely imaginary.

We observe that $\omega_+$ has only a branch cut in $\gamma_1$, furthermore, since $\omega_+^2-\omega_-^2= 4\mu \sqrt{Z}$, its branch cut is $\gamma_1$.   
The points where $\omega_-$ vanishes correspond to the extremal points of $\gamma_2$, those where $\omega_+$ vanishes are the extremal points of $\gamma_1$ and those where $(\omega_+^2-\omega_-^2)$ is zero the extremal points of $\gamma_1$. 

\begin{figure}[h!]
	\begin{center}
     \begin{picture}(350,180)(0,0) 
       \put(15,-30){\includegraphics[width=110mm]{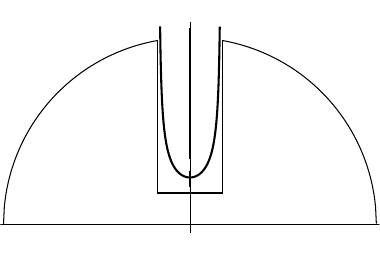}}
        \put(175,-9){\large $0$}
        \put(100,-9){\large $\xi_1$}
        \put(75,133){\large $\xi_2$}
        \put(133,85){\large $\xi_3$}
        \put(175,15){\large $\xi_4$}
        \put(202,85){\large $\xi_5$}

        \put(175,150){$\gamma_1$}
        \put(175,32){$\gamma_2$}
        \put(149,150){$\gamma_3$}
     \end{picture}
	\end{center}
 \caption{The tick lines correspond to the union of the brunch cuts of $\omega_-$  $\{\gamma_i\}$ while the thin line is the contour of integration $\{\xi_i\}$. }	
\label{fig:cuts}
\end{figure}
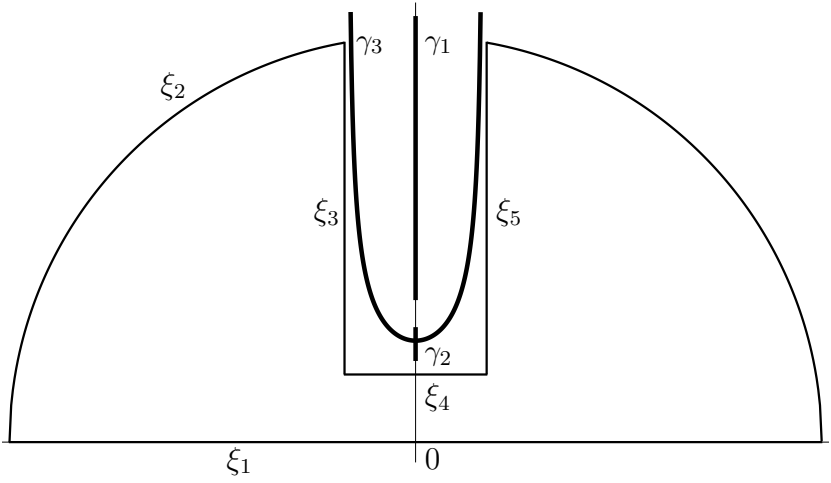
\bigskip
The $p_1$ integration can evaluated by choosing a contour in the upper half plane. The contour can the taken in such a way that both the the poles and the branch cuts lie outside it. Furthermore, the non trivial part of the contour is chosen to be at imaginary part larger than $\sqrt{|\mathbf{p}_\perp|^2+M^2-\delta M^2}$.
The contour $\xi$ is formed by the following paths $\xi_1,\xi_2, \xi_3, \xi_4,\xi_5 $ where
 $\xi_1$ is the real line,
$\xi_2$ is  part of the semicircle centered  in the origin of the complex plane and with radius $B$ tending to infinity. Furthermore 
the points with $|\text{Re}z|<2\mu$ are removed from the circle. 
\begin{align*}
\xi_1 &= \{ x \;|\;   -B<x<B\}, \\
\xi_2 &= \{x+iy \;|\;   x^2+y^2=B^2, y\geq 0 , |x| > \mu\}, \\
\xi_3 &= \{2\mu +iy \;|\;   \sqrt{|\mathbf{p}_\perp|^2+M^2 -\delta M^2}-\epsilon<y<B\}, \\
\xi_4 &= \{x +i\sqrt{|\mathbf{p}_\perp|^2+M^2 -\delta M^2} -i\epsilon \;|\;  -2\mu<x<2\mu\}, \\
\xi_5 &= \{-2\mu +iy \;|\;   \sqrt{|\mathbf{p}_\perp|^2+M^2 -\delta M^2}-\epsilon<y<B\}.
\end{align*}
Notice that, the poles and the branch cuts lies outside of this contour.
Furthermore, in the limit $B \to 0$ the integral done on $\xi_2$ vanishes for every value of $\mathbf{p}_\perp$. 
Hence, in view of the residue theorem, and in the limit $B\to\infty$ 
the integral over $\xi_1$, which is nothing but $I^2_{\sigma}$, equals the sum of the integrals over $\xi_U=\xi_3\cup\xi_4\cup\xi_5$ with $B\to \infty$.

\bigskip
To prove that the limit $B\to\infty$ of the integral over $\xi_2$ vanishes we observe that
\begin{align*}
\omega_\sigma^2  =  w^2 + 2\mu^2 +\sigma  2\mu \sqrt{w^2 +\mu^2 + \frac{\delta M^4}{4\mu^2}} 
=  \left(\sqrt{w^2 +\mu^2 + \frac{\delta M^4}{4\mu^2}} + \sigma \mu \right)^2 -\frac{\delta M^4}{4\mu^2}
\end{align*}
hence
\begin{align*}
|\omega_\sigma| 
&\leq    \left|\sqrt{w^2 +\mu^2 + \frac{\delta M^4}{4\mu^2}} + \sigma \mu \right| +\frac{\delta M^2}{2\mu}\\
&\leq    \sqrt{ \left|w^2 +\mu^2 + \frac{\delta M^4}{4\mu^2}\right|} + \mu  +\frac{\delta M^2}{2\mu}\\
&\leq    \sqrt{ x^2 + y^2 + |\mathbf{p}_\perp|^2 + M^2+ \mu^2 + \frac{\delta M^4}{4\mu^2}} + \mu  +\frac{\delta M^2}{2\mu}\\
&\leq    \sqrt{x^2 + y^2+ |\mathbf{p}_\perp|^2 + M^2-\delta M^2 } + 2\mu  +\frac{\delta M^2}{\mu} +\delta M
\end{align*}
Furthermore, for large values of $|z|$ if we stay outside of the region where are the branch cuts are located we have that
\[
\frac{1}{|\omega_\sigma|}\leq \frac{C}{|z|}, \qquad  \frac{1}{|\omega_+^2-\omega_-^2|}\leq \frac{C}{|z|}.
\]
Hence, if $r$ is sufficiently large,  the exponential growth in the estimate of $\hat{f}(\omega,p)$ is controlled by $|e^{ip_1r}|\leq |e^{-ry }| $, furthermore 
\[
\frac{1}{|\omega_\sigma||\omega_+^2-\omega_-^2|}
\]
is bounded towards $\xi_2$, hence, in the limit $B\to\infty$, the integral over $\xi_2$ vanishes.

\bigskip
It remains to analyze the contribution of $\xi_U$. We need the following estimates.
In particular, in the region contained in $\cup_i\xi_i$, the real part of any square root is positive, we have that 
\[
|\omega_+^2| \geq |\text{Im} \omega_+^2 | = 2|xy| \left| 1 + \frac{\sqrt{2}\mu}{\sqrt{\sqrt{a^2+b^2}+a}} \right|  \geq  2|xy|, \qquad z\in \xi_U
\]
where $a$ and $b$ are given in \eqref{eq:a-b}.
For $z\in\xi_4$, 
since there $-y^2 + |\mathbf{p}_\perp|^2 + M^2>0$,
it holds that 
\[
|\omega_+^2| \geq \left|\text{Re} \left( w^2+2\mu^2+2\mu \sqrt{w^2 +\mu^2 +\frac{\delta M^4}{4\mu^2}} 
\right)  \right| \geq 2\mu \sqrt{\left|\text{Re} \left( w^2 +\mu^2 +\frac{\delta M^4}{4\mu^2}\right) \right|} \geq 2\mu, \qquad z\in \xi_4.
\]
Furthermore
\[
|\omega_+^2-\omega_-^2| \geq \sqrt{|\text{Im} (\omega_+^2-\omega_-^2)^2 |} = \sqrt{2|xy|}, \qquad z\in \xi_U
\]
and on $\xi_4$, since there $-y^2 + |\mathbf{p}_\perp|^2 + M^2+2\mu^2>0$, it holds that 
\[
|\omega_+^2-\omega_-^2|^2 \geq \left|\text{Re} \left(w^2+\mu^2+ \frac{\delta M^4}{4\mu^2}\right) \right| = \left|x^2-y^2 + |\mathbf{p}_\perp|^2 + M^2+2\mu^2\right|\geq \mu^2, \qquad z\in \xi_4
\]
On $\xi_4$ we have that $a$ given in \eqref{eq:a-b} is such that $a>\frac{1}{4\mu^2}(\delta M^2+2\mu^2)^2$, hence
\[
|\omega_-^2|\geq |\text{Im}(\omega_-^2)|  \geq  
2|xy| \left|  1  - \frac{\sqrt{2}\mu}{\sqrt{\sqrt{a^2+b^2}+a}} \right|  \geq  2|xy| \left|  1  - \frac{\mu}{\sqrt{a}} \right| = 2|xy| \frac{\delta M^2}{
\delta M^2 + 2\mu^2} 
, \qquad z\in \xi_4
\]
furthermore, on $\xi_3\cup\xi_5$ 
\[
\sqrt{\sqrt{a^2+b^2}+a} \geq \sqrt{2}\mu
\]
and it is a monotonically decreasing function of $y$. 
Its supremum on $\xi_3\cup\xi_4$ is reached at $y^2=|\mathbf{p}_\perp|^2+M^2-\delta M^2 -\epsilon$ and there $a\geq \mu^2+(\mu +\frac{\delta M^2}{2\mu})^2\geq 2\mu^2$ and $b\geq 2\mu\sqrt{\mu^2+(\mu +\frac{\delta M^2}{2\mu})^2} \geq 2\sqrt{2} \mu^2$, hence 
\[
\sup_{\xi_3\cup\xi_4}{\sqrt{a^2+b^2}+a} \geq (\sqrt{12}+2)\mu^2.
\]
Hence, for any $\frac{1}{2} < l^2 < \sqrt{\sqrt{3}+1}-1$
we can find a $\tilde{y}$ where $\sqrt{\sqrt{a^2+b^2}+a}=\sqrt{2}\mu(1+l^2)$ and
\[
\text{Re} (-\omega_-^2) = y^2-x^2 -|\mathbf{p}_\perp|^2-M^2-2\mu^2+2\mu \sqrt{2}\sqrt{\sqrt{a^2+b^2}+a} \geq  y^2-x^2 -|\mathbf{p}_\perp|^2-M^2 +2\mu^2 l^2, \qquad  y\geq \tilde{y}
\]
\[
\text{Re} (-\omega_-^2) \geq -\mu^2 +\delta M^2 +2\mu^2 l^2 \geq (l^2-1) \mu+ \delta M^2, \qquad  y\geq \tilde{y} 
\]
\[
|\text{Im} (\omega_-^2)| = |2xy| \left| 1 - \frac{\sqrt{2}\mu}{\sqrt{\sqrt{a^2+b^2}+a}}\right|  \geq  |2xy| \frac{l^2}{l^2+1} 
\qquad  y\leq \tilde{y}
\]
\[
|\text{Im} (\omega_-^2)| \geq 
 |4\mu \sqrt{M^2-\delta M^2 -\epsilon}| 
\qquad  y\leq \tilde{y}
\]

\bigskip
Combining all these estimates we have that on $\xi_3$ and $\xi_5$  
\[
\left|\frac{1}{(\omega_+^2-\omega_-^2)\omega_\sigma}\right| \leq C.
\]
uniformly in $\mathbf{p}_\perp$ and $\epsilon$ and the same holds true for  $\frac{1}{(\omega_+^2-\omega_-^2)\omega_+}$ and for  $\frac{1}{(\omega_+^2-\omega_-^2)}$ on $\xi_4$ while $\frac{1}{\omega_-}$ is bounded by an $L^1$ function uniformly in $\epsilon$ and $\mathbf{p}_\perp$ on $\xi_4$.  Furhtermore
\[
\frac{1}{|\omega_-|} \leq |E|, \qquad z\in\xi_4.
\]

\bigskip
With this observation we can now control the integrals over $\xi_U$. As an example, consider the contribution to $I^2_{-}(u,r)$ due to $I^2_{-,\xi_5}(u,r)$
\begin{align*}
|I^2_{-,\xi_5}(u,r)| 
&\leq
\frac{1}{(2\pi)^3} \int d\mathbf{p}_\perp   \int_{\sqrt{|\mathbf{p}_\perp|^2+M^2-\delta M^2}-\epsilon}^\infty dy  \left|\frac{1}{(\omega_+^2-\omega_-^2)\omega_-}\right|
e^{-r y}
e^{R  \sqrt{2}(\sqrt{|\mathbf{p}_\perp|^2+M^2-\delta M^2}+|y|)} 
\\
&\leq
\frac{e^{r\epsilon}}{(2\pi)^3} \int d\mathbf{p}_\perp   \int_{0}^\infty dy   |C|
e^{-(r-2\sqrt{2} R ) (\sqrt{|\mathbf{p}_\perp|^2+M^2-\delta M^2}+|y| )}
\\
&\leq
e^{r\epsilon}e^{-(r-2\sqrt{2} R )  \sqrt{ M^2 - \delta M^2}}
\frac{1}{(2\pi)^3} \int d\mathbf{p}_\perp   \int_{0}^\infty dy   |C|
e^{-(r-2\sqrt{2} R ) (\sqrt{|\mathbf{p}_\perp|^2+M^2-\delta M^2}-\sqrt{ M^2 - \delta M^2}+|y|)}
\\
&\leq
e^{r\epsilon}e^{-(r-2\sqrt{2} R )  \sqrt{ M^2 - \delta M^2}}
\frac{1}{(2\pi)^3} \int d\mathbf{p}_\perp   \int_{0}^\infty dy   |C|
e^{-(3-2\sqrt{2})R  (\sqrt{|\mathbf{p}_\perp|^2+M^2-\delta M^2}-\sqrt{ M^2 - \delta M^2}+|y|)}
\end{align*}
where we used \eqref{eq:paley-wiener} and 
 in the last inequality we used the fact that $r>3Rc$ and the inequalities holds uniformly in $u$.
Hence the integral can be taken and it can be bounded by a constant uniformly in $r$ to get that 
\begin{align*}
|I^2_{-,\gamma_5}(u,r)| 
&\leq C
e^{-r  (\sqrt{ M^2 - \delta M^2}-\epsilon)}
\end{align*}
and the constant does not depend on $r$ or $\epsilon$ for $r > 3R$. Since the inequality holds for every $\epsilon$ we have that 
\begin{align*}
|I^2_{-,\gamma_5}(u,r)| 
&\leq C
e^{-r  \sqrt{ M^2 - \delta M^2}}
\end{align*}
The estimate of $I^2_{-,\gamma_3}$ can be done in the same way. The analysis of $I^2_{-,\gamma_4}$ can also be done in the same way substituting $|C|$ with the bounding $L^1$ function $|E|$. We actually have
 \begin{align*}
|I^2_{-,\xi_4}(u,r)| 
&\leq
\frac{1}{(2\pi)^3} \int d\mathbf{p}_\perp   \int_{-\mu}^\mu dx  \left|\frac{1}{(\omega_+^2-\omega_-^2)\omega_-}\right|
e^{-r \sqrt{|\mathbf{p}_\perp|^2+M^2-\delta M^2}-\epsilon}
e^{R \sqrt{2}(\sqrt{|\mathbf{p}_\perp|^2+M^2-\delta M^2}+|x|)} 
\\
&\leq
\frac{e^{r\epsilon}}{(2\pi)^3} \int d\mathbf{p}_\perp   \int_{-\mu}^\mu dx   |E|
e^{-(r-2\sqrt{2} R ) (\sqrt{|\mathbf{p}_\perp|^2+M^2-\delta M^2}+\mu )}
\\
&\leq
e^{r\epsilon}e^{-(r-2\sqrt{2} R )  \sqrt{ M^2 - \delta M^2}} C' 
\int d\mathbf{p}_\perp   
e^{-(3-2\sqrt{2})R  (\sqrt{|\mathbf{p}_\perp|^2+M^2-\delta M^2}-\sqrt{ M^2 - \delta M^2})}
\end{align*}
and it can be treated in the same way as before. 
All the other contributions can be analysed in a similar way.
\end{proof}


\begin{thebibliography}{999}


\bibitem[ABS08]{Schmitt} 
  M.~G.~Alford, M.~Braby and A.~Schmitt,
{\it ``Critical temperature for kaon condensation in color-flavor locked quark matter,''}
  J.\ Phys.\ G {\bf 35}, 025002 (2008)

\bibitem[Al90]{Altherr}
T.~Altherr {\it ``Infrared problem in $g\phi^4$ theory at finite temperature,''}
Phys.\ Lett.\ {\bf B 238}, 360 (1990)


\bibitem[And95]{Anderson}
M.~H.~Anderson, J.~R.~Ensher, M.~R.~Matthews, C.~E.~Wieman and E.~A.~Cornell,
{\it ``Observation of bose-einstein condensation in a dilute atomic vapor,''}
Science {\bf 269}, 198 (1995)

\bibitem[Ar73]{Araki-KMS}
H.~Araki,
{\it ``Relative Hamiltonian for faithful normal states of a von Neumann algebra,''} 
Publ.\ RIMS, Kyoto Univ.\ {\bf 9}(1), 165-209 (1973).

\bibitem[Be94]{Benfatto}
G.~Benfatto
{\it ``Renormalization group approach to zero temperature Bose condensation,''}
Constructive Physics Results in Field Theory, Statistical Mechanics and Condensed Matter Physics,
Lecture Notes in Physics vol. 446
Edited by V. Rivasseau, Springer (1994)

\bibitem[BN77]{QGP}
H.~Bohr, H.~B.~Nielsen
{\it ``Hadron production from a boiling quark soup: A thermodynamical quark model predicting particle ratios in hadronic collisions,''}
Nuc. Phys. {\bf  B128} 275-293 (1977)

\bibitem[BS76]{BS}
N.~N.~Bogoliubov and D.~V.~Shirkov,
{\it ``Introduction to the Theory of Quantized Fiedls,}New York: John Wiley and Sons, 1976.

\bibitem[BMZ16]{BosonStar}
E.~Braaten, A.~Mohapatra and H.~Zhang,
{\it ``Dense Axion Stars,''}
Phys. Rev. Lett. {\bf 117}, 121801 (2016)

\bibitem[Bra95]{Bradely}
C.~C.~Bradley,  C.~A.~Sackett and J.~J.~Tollett,  
{\it ``Evidence of Bose-Einstein Condensation in an Atomic Gas with Attractive Interactions,''}
Phys. Rev. Lett. {\bf 75}, 1687 (1995)


\bibitem[BR97]{BRvol2}
O.~Bratteli,  D.W.~Robinson, 
{\it ``Operator algebras and quantum statistical mechanics 2,''} 
Springer, Berlin (1997).


\bibitem[BB98]{BrosBuchholz}
J.~Bros and D.~Buchholz,
{\it ``The unmasking of thermal Goldstone bosons,''}
Phys. Rev. {\bf D58}, 125012 (1998)

\bibitem[BB94]{BrosBuchholz1}
J.~Bros and D.~Buchholz,
{\it ``Towards a Relativistic KMS Condition,''}
Nucl.\ Phys. {\bf B429}, 291-318 (1994)

\bibitem[BDF09]{BDF09}
  R.~Brunetti, M.~Duetsch and K.~Fredenhagen,
 {\it ``Perturbative Algebraic Quantum Field Theory and the Renormalization Groups,''}
  Adv.\ Theor.\ Math.\ Phys.\  {\bf 13}, 1541 (2009)

\bibitem[BF00]{BF00} 
  R.~Brunetti and K.~Fredenhagen,
{ \it ``Microlocal analysis and interacting quantum field theories: Renormalization on physical backgrounds,''}
  Commun.\ Math.\ Phys.\  {\bf 208}, 623 (2000)


\bibitem[BFK96]{BFK} 
  R.~Brunetti, K.~Fredenhagen and M.~K\"ohler,
{\it ``The Microlocal spectrum condition and Wick polynomials of free fields on curved space-times,''}
  Commun.\ Math.\ Phys.\  {\bf 180}, 633 (1996)


\bibitem[BFV03]{BFV03} 
  R.~Brunetti, K.~Fredenhagen and R.~Verch,
{\it ``The Generally covariant locality principle: A New paradigm for local quantum field theory,''}
  Commun.\ Math.\ Phys.\  {\bf 237}, 31 (2003)

\bibitem[Bu18]{Buchholz-Interacting-Bosons}
D.~Buchholz, {\it ``The resolvent algebra of non-relativistic Bose fields: observables, dynamics and states,''}
Commun.\ Math.\ Phys.\ {\bf 362}, 949-981 (2018)


\bibitem[BG08]{Buchholz-Grundling}
D.~Buchholz and H.~Grundling,
{\it ``The resolvent algebra: A new approach to canonical quantum systems,''}
J.\ Func.\ Anal. {\bf 254}, 2725-2779 (2008)

\bibitem[BR14]{BuchholzRoberts}
D.~Buchholz and J.~E.~Roberts
{\it ``New Light on Infrared Problems: Sectors, Statistics, Symmetries and Spectrum,''}
Commun.\ Math.\ Phys.\  {\bf 330}, 935 (2014)

\bibitem[CF09]{ChilianFredenhagen}
B. Chilian and K. Fredenhagen,
{\it ``The time slice axiom in perturbative quantum field theory on globally hyperbolic spacetimes,''}
Commun.\ Math.\ Phys.\ {\bf 287}, 513-522 (2009)


\bibitem[DF98]{DuetschFredenhagen}
M. Duetsch and K. Fredenhagen
{\it ``A local (perturbative) construction of observables in gauge theories: the example of QED,''}
Commun.\ Math.\ Phys.\ {\bf 203}, 71-105 (1999)




\bibitem[Dav95]{Davies}
K.~B.~Davis, M.~-O.~Mewes, M.~R.~Andrews, N.~J.~van~Druten, D.~S.~Durfee, D.~M.~Kurn and W.~Ketterle,
{\it ``Bose-Einstein Condensation in a Gas of Sodium Atoms,''}
Phys. Rev. Lett. {\bf 75}, 3969 (1995)



\bibitem[Dr19]{Drago}
N.~Drago, {\it ``Thermal State with Quadratic Interaction,''}
Annales Henri Poincar\'e {\bf 20}, 905-927 (2019)
    
    

\bibitem[DHP17]{DHP} 
  N.~Drago, T.~P.~Hack and N.~Pinamonti,
{\it ``The generalised principle of perturbative agreement and the thermal mass,''}
  Annales Henri Poincar\'e {\bf 18}, 807 (2017)

\bibitem[Du19]{Duetsch}
M. D\"utsch
{\it ``From Classical Field Theory to Perturbative Quantum Field Theory,''}
Springer Berlin 2019

\bibitem[DF03]{DF03}
M. D\"utsch and K. Fredenhagen, 
{\it ``The master Ward identity and generalized Schwinger-Dyson equation in classical field theory,''} 
Commun. Math. Phys. {\bf 243}, 275 (2003).


\bibitem[DF04]{DF04}
M. D\"utsch and K. Fredenhagen, 
{\it ``Causal perturbation theory in terms of retarded products, and a proof of the action Ward identity,''} 
Rev. Math, Phys. {\bf 16}, 1291-1348 (2004).
 


\bibitem[EG73]{EG}
H.~Epstein and V.~Glaser, 
{\it ``The role of locality in perturbation theory, ''}
Ann.\ Inst.\ Henri\ Poincar\'e Section
A, vol. XIX, n.3, {\bf 211} (1973)


\bibitem[FPB82]{FPV}
M. Fannes, J. V. Pul\`e, and V. Verbeure, 
{\it ``On Bose condensation,''} 
Helv. Phys. Acta {\bf 55}, 391-399 (1982)

\bibitem[FL14]{FredenhagenLindner} 
  K.~Fredenhagen and F.~Lindner,
  {\it ``Construction of KMS States in Perturbative QFT and Renormalized Hamiltonian Dynamics,''}
  Commun.\ Math.\ Phys.\  {\bf 332}, 895 (2014)

\bibitem[FK13]{FredenhagenRejznerBV} 
K.~Fredenhagen and K.~Rejzner,
{\it ``Batalin-Vilkovisky Formalism in Perturbative Algebraic Quantum Field Theory,''}
Commun.\ Math.\ Phys.\ {\bf 317}, 697-725 (2013)



\bibitem[FK15]{FredenhagenRejzner} 
  K.~Fredenhagen and K.~Rejzner,
{\it ``Perturbative algebraic quantum field theory,''}
In: Calaque D., Strobl T. (eds) Mathematical Aspects of Quantum Field Theories. Mathematical Physics Studies. Springer, Cham
(2015)
 
\bibitem[FR75]{Friedlander} Friedlander, {\it ``The wave equation on a curved space-time''} Cambridge University Press (1975)

\bibitem[GSW62]{GoldstoneSalamWeinberg}
J.~Goldstone, A.~Salam and S.~Weinberg
{\it ``Broken Symmetries,''}
Phys.\ Rev.\ {\bf 127}, 965 (1962)

\bibitem[Gr61]{Gross}
E.~P.~Gross, {\it ``Structure of a quantized vortex in boson systems,''} 
Il Nuovo Cimento {\bf 20}, 454-457 (1961)



\bibitem[Ha92]{Haag}
R.~Haag,
{\it ``Local quantum physics,''}
Second edition, Springer-Verlag Berlin, 1992
ISBN: 3-540-61451-6.

\bibitem[HHW67]{KMS}
R.~Haag, N.~Hugenholtz and M.~Winnink,
{\it ``On the equilibrium state in quantum statistical mechanics,''}
Commun.\ Math.\ Phys.\ {\bf 5}, 215 (1967)

\bibitem[HK64]{HaagKastler}
R.~Haag and D.~Kastler, {\it ``An algebraic approach to quantum field theory,''} 
J. Math. Phys., {\bf 5}, 848 (1964)



\bibitem[Ho07]{Hollands} 
  S.~Hollands,
 {\it ``Renormalized Quantum Yang-Mills Fields in Curved Spacetime,''}
Rev.\ Math.\ Phys. {\bf 20}, 1033-1172 (2008).


\bibitem[HW01]{HW01} 
  S.~Hollands and R.~M.~Wald,
 {\it ``Local Wick polynomials and time ordered products of quantum fields in curved space-time,''}
  Commun.\ Math.\ Phys.\  {\bf 223}, 289 (2001)


\bibitem[HW02]{HW02} 
  S.~Hollands and R.~M.~Wald,
 {\it ``Existence of local covariant time ordered products of quantum fields in curved space-time,''}
  Commun.\ Math.\ Phys.\  {\bf 231}, 309 (2002)

\bibitem[HW05]{HW05} 
S.~Hollands and R.~M.~Wald,
{\it ``Conservation of the stress tensor in interacting quantum field theory in curved spacetimes,''}
Rev.\ Math.\ Phys.\  {\bf 17}, 227 (2005)

\bibitem[Ho03]{Hormander}
L.~H\"ormander
{\it ``The Analysis of Linear Partial Differential Operators I,''}
Springer Berlin (2003)



\bibitem[JW11]{JakelWreszinski}
C.~D.~J\"akel and W.~F.~Wreszinski,
{\it ``A Goldstone theorem in thermal relativistic quantum field theory,''}
J.\ Math.\ Phys. {\bf 52}, 012302 (2011)

\bibitem[Jo64]{JonaLasinio}
G.~Jona-Lasinio
{\it ``Relativistic Field Theories with Symmetry-Breaking Solutions,''}
Nuovo Cimento {\bf 34}, 1790 (1964)


\bibitem[KRS66]{KRS}
D.~Kastler, D.~W.~Robinson and A.~Swieca, 
{\it ``Conserved Currents and Associated Symmetries; Goldstone's Theorem,''} 
Commun.\ Math.\ Phys. {\bf 2},  108--120 (1966)



\bibitem[KW91]{KayWald}
B.~S.~Kay and R.~M.~Wald,
{\it ``Theorems on the uniqueness and thermal properties of stationary, nonsingular, quasifree states on spacetimes with a bifurcate killing horizon,''}
Phys.\ Rep.\ {\bf 207}, 49-136 (1991)


\bibitem[Ko87]{Kowalski}
K.~L.~Kowalski
{\it ``Goldstone theorem at finite temperature and density,''}
Phys.\ Rev.\ {\bf D35}, 3940 (1987)





\bibitem[LSY01a]{LiebSeiringerYngvasonCMP}
E.~Lieb, R.~Seiringer, J.~Yngvason,
{\it ``A Rigorous Derivation of the Gross-Pitaevskii Energy Functional for a Two-dimensional Bose Gas''}
Commun. Math. Phys. {\bf 224}, 17 (2001) 

\bibitem[LSY01b]{LiebSeiringerYngvason}
E.~Lieb, R.~Seiringer, J.~Yngvason 
{\it Bosons in a trap: A rigorous derivation of the Gross-Pitaevskii energy functional.} (2001)  In: Thirring W. (eds) The Stability of Matter: From Atoms to Stars. Springer, Berlin, Heidelberg

\bibitem[LSY05]{LiebSeiringerYngvason05}
E.H.~Lieb, R.~Seiringer, J.~Yngvason, 
{\it ``Justification of c-Number Substitutions in Bosonic Hamiltonians,''}
Phys. Rev. Lett. {\bf 94}, 080401 (2005)


\bibitem[LSS05]{LiebSeiringerSolovejYngvason}
E.~Lieb, R.~Seiringer, J.~P.~Solovej, J.~Yngvason 
{\it The Mathematics of the Bose Gas and its Condensation.} (2005)
Birkh\"auser,  Basel 
%





\bibitem[Li13]{Lindner}
F.~Lindner,
{\it ``Perturbative Algebraic Quantum Field Theory at Finite Temperature,''}
PhD thesis, University of Hamburg (2013). 


\bibitem[MS87]{MorchioStrocchi}
G.~Morchio and F.~Strocchi,
{\it ``Mathematical structures for long?range dynamics and symmetry breaking,''}
J.\ Math.\ Phys.\ {\bf 28}, 622 (1987)

\bibitem[Mo03]{Moretti}
  V.~Moretti,
{\it ``Comments on the stress-energy tensor operator in curved spacetime,''}
  Commun.\ Math.\ Phys.\  {\bf 232}, 189 (2003).


\bibitem[PCDS04]{DiCastro}
F.~Pistolesi, C.~Castellani, C.~Di~Castro and G.~C.~Strinati
{\it ``Renormalization-group approach to the infrared behavior of a zero-temperature Bose system,''}
Phys.\ Rev.\ {\bf B 69}, 024513 (2004)

\bibitem[Pi61]{Pitaevskii}
L.~P.~Pitaevskii,
{\it ``Vortex lines in an imperfect Bose gas,''}
Sov. Phys. JETP. {\bf 13}, 451-454 (1961)

\bibitem[PS16]{Pista}
L.~P.~Pitaevskii, S. Stringari,
{\it ``Bose-Einstein Condensation and Superfluidity''} (2016), International Series of Monographs in Physics, 164, OUP Oxford



\bibitem[Ra96]{Radzikowski} 
  M.~J.~Radzikowski,
{\it ``Micro-local approach to the Hadamard condition in quantum field theory on curved space-time,''}
  Commun.\ Math.\ Phys.\  {\bf 179}, 529 (1996)

\bibitem[Re16]{Rejzner}
K. Rejzner
{\it ``Perturbative Algebraic Quantum Field Theory, An Introduction for Mathematicians''}
Springer Berlin 2016


\bibitem[RST70]{RoccaSirugueTestard}
F. Rocca, M. Sirugue and D. Testard
{\it ``On a Class of Equilibrium States under the Kubo-Martin-Schwinger Condition. II. Bosons,''}
Commun.\ Math.\ Phys.\  {\bf 19}, 119-141 (1970)


\bibitem[Sa11]{Satz}
H.~Satz {``The Quark-Gluon Plasma,''}
Nucl.~Phys.~{\bf A862-863}, 4-12 (2011)

\bibitem[Sc89]{Scharf}
G.~Scharf, 
{\it ``Finite Quantum Electrodynamics''}
Springer Berlin 1989


\bibitem[So13]{LightCondensation} 
D.~N.~Sob'yanin, 
{\it ``Bose-Einstein condensation of light: General theory,''}
Phys. Rev. E {\bf 88}, 022132 (2013)

\bibitem[St95]{Steinmann}
O.~Steinmann,
{\it ``Perturbative Quantum Field Theory at Positive Temperatures: An Axiomatic Approach,''}
Commun.~Math.~Phys.~{\bf 170}, 405-415 (1995)

\bibitem[St08]{Strocchi} 
F.~Strocchi
{\it ``Symmetry Breaking,''}
Lecture Notes in Physics (2008) 	
Springer Berlin Heidelberg

\bibitem[St51]{Stueckelberg} 
E.~C.~G.~Stueckelberg, 
{\it ``Relativistic Quantum Theory for Finite Time Intervals,''}
Phys.\ Rev.\   {\bf 81}, 130 (1951)

\bibitem[SR49]{Stueckelberg2} 
E.~C.~G.~Stueckelberg and D. Rivier, 
{\it ``Causalit\'e et structure de la Matrice S,''}
Helv.\ Phys.\ Acta, {\bf 23}, 216 (1949)

\bibitem[S{\"{u}}05]{Suto}
A. S\"ut\H{o}, 
{\it ``Equivalence of Bose-Einstein Condensation and Symmetry Breaking,''} 
Phys. Rev. Lett. {\bf 94}, 080402 (2005)

\bibitem[Sw67]{Swieca} 
J.~A.~Swieca
{\it ``Range of Forces and Broken Symmetries in Many-Body Systems,''}
Commun.\  Math.\ Phys. {\bf 4}, 1-7 (1967)

\bibitem[Wa78]{Wald}
R.~M.~Wald,
  {\it ``Trace Anomaly Of A Conformally Invariant Quantum Field In Curved Space-Time,''}
  Phys.\ Rev.\  {\bf D 17}, 1477 (1978).


\end{thebibliography}
\end{document}